\documentclass{amsart}

\usepackage{amsmath}
\usepackage{amssymb}
\usepackage{amsthm}
\usepackage{amscd}

\newtheorem{thm}{Theorem}[section]
\newtheorem{prop}[thm]{Proposition}
\newtheorem{lem}[thm]{Lemma}
\newtheorem{cor}[thm]{Corollary}
\newtheorem{claim}[thm]{Claim}
\newtheorem{assumption}[thm]{Assumption}

\theoremstyle{definition}
\newtheorem{dfn}[thm]{Definition}

\theoremstyle{remark}

\newtheorem*{acknowledgment}{Acknowledgements}

\newcommand{\C}{\mathbb{C}}
\newcommand{\N}{\mathbb{N}}
\newcommand{\R}{\mathbb{R}}
\newcommand{\Z}{\mathbb{Z}}

\title[Homological bulk-edge correspondence for Weyl semimetal]
{Homological bulk-edge correspondence for \\ Weyl semimetal}

\author[K. Gomi]{Kiyonori Gomi}

\address{
Department of Mathematics, 
Tokyo Institute of Technology, 
2-12-1 Ookayama, Meguro-ku, Tokyo, 152-8551, Japan.}

\email{kgomi@math.titech.ac.jp}



\date{}

\begin{document}

\begin{abstract}
For a certain translation invariant tight-binding model of three-dimensional Weyl semimetals, we establish a bulk-edge correspondence as an equality of two relative homology classes, based on an idea of Mathai and Thiang: From spectral information on the edge Hamiltonian, we construct a relative homology class on the surface momentum space. This class agrees with the image under the surface projection of a homology class on the bulk momentum space relative to the Weyl points, constructed from the bulk Hamiltonian. Furthermore, the relative homology class on the surface momentum space can be represented by homology cycles whose images constitute the Fermi arc, the locus where the edge Hamiltonian admits zero spectrum.
\end{abstract}

\maketitle

\tableofcontents


\section{Introduction}
\label{sec:introduction}

\subsection{Bulk-edge correspondence for topological insulators}

The \textit{bulk-edge correspondence} is a prominent feature of topological insulators. As is known, in the experimental observation of the $3$-dimensional quantum spin Hall effect \cite{Hsieh}, the non-trivial topological phases of the bulk insulating states are confirmed by detecting the corresponding edge states. Thus, from theoretical viewpoint, the justification of the bulk-edge correspondence is an important theme, and has been achieved in various situations as the coincidence of two quantities called the bulk index and the edge index (see e.g.\ \cite{ASV,ASS,BES,Bra,G-P,Hatsu,KRS1,KRS2,P-SB}).

\medskip

For example, in a tight-binding model which describes an insulator as a quantum system on the $d$-dimensional square lattice $\Z^d$ with translation invariance, the information on the Hamiltonian $H$ of the quantum system is encoded into a continuous family $\hat{H} = \{ \hat{H}(k) \}_{k \in T^d}$ of invertible Hermitian matrices parametrized by the quasi-momentum $k$, which forms the $d$-dimensional torus $T^d = \R^d/2\pi \Z^d$ called the Brillouin zone. The eigenspaces of $\hat{H}(k)$ with negative eigenvalues form a vector bundle $E_{\hat{H}}$ on $T^d$ called the Bloch bundle. In the case that $d = 2$, we have the first Chern number of $E_{\hat{H}}$, which we denote by $c_1(\hat{H}) \in \Z$. This is the bulk index. A typical Hamiltonian, called the Qi-Wu-Zhang (QWZ) model in \cite{AOP}, is
$$
\hat{H}_{\mathrm{QWZ}}(k_x, k_y)
= \sin k_x \sigma k_x + \sin k_y \sigma_y + 
(u + \cos k_x + \cos k_y) \sigma_z,
$$
where $\sigma_x$, $\sigma_y$ and $\sigma_z$ are the Pauli matrices
\begin{align*}
\sigma_x &=
\left(
\begin{array}{cc}
0 & 1 \\
1 & 0
\end{array}
\right),
&
\sigma_y &=
\left(
\begin{array}{rr}
0 & -i \\
i & 0
\end{array}
\right),
&\sigma_z &=
\left(
\begin{array}{rr}
1 & 0 \\
0 & -1
\end{array}
\right).
\end{align*}
If we take the parameter $u$ to be $0 < u < 2$, then $\hat{H}_{\mathrm{QWZ}}(k_x, k_y)$ is invertible for all $(k_x, k_y) \in T^2$, and the Chern number is $c_1(\hat{H}_{\mathrm{QWZ}}) = 1$.

To formulate the edge index, one introduces an edge to the lattice, and then impose the Dirichlet boundary condition to define a self-adjoint operator $H^\sharp$, often called the edge Hamiltonian. Using the translation invariance in the direction along the edge, $H^\sharp$ induces a continuous family of self-adjoint operators $\acute{H}^\sharp = \{ \acute{H}^\sharp(k) \}_{k \in S^1}$ parametrized by the circle $T^1 = S^1 = \R/2\pi \Z$. It turns out that $\acute{H}^\sharp(k)$ is Fredholm, so that the spectral flow $\mathrm{sf}(\acute{H}^\sharp) \in \Z$ makes sense. The edge index is $- \mathrm{sf}(\acute{H}^\sharp)$. Now, the bulk-edge correspondence for the $2$-dimensional topological insulator (Chern insulator) is stated as the equality $c_1(\hat{H}) = -\mathrm{sf}(\acute{H}^\sharp)$.

\subsection{Bulk-edge correspondence for Weyl semimetals}

The bulk-edge correspondence is also prominent in other topological states of matters. This is the case for the \textit{Weyl semimetal} \cite{AMV,Wan,Xu}, which is experimentally confirmed by observing its characteristic edge states, known as \textit{Fermi arcs}. However, mathematical study of the bulk-edge correspondence for Weyl semimetals, in particular one which explains topology of Fermi arcs, seem to be not fully developed yet. In this direction of research, a work of Mathai and Thiang \cite{M-T1,M-T2} suggests a formulation of the bulk-edge correspondence. Based on their idea, we mathematically formulate and prove a bulk-edge correspondence for a particular model of Weyl semimetals in this paper. The model is more or less a protypical one, so that our result should be regarded as a first step toward the general case. 
\medskip

Based on the idea in \cite{M-T1,M-T2}, we here formulate the bulk-edge correspondence that we anticipate. Let $\hat{H} : T^3 \to \mathrm{Herm}(\C^2)_0$ be a continuous map from the $3$-dimensional torus to the space $\mathrm{Herm}(\C^2)_0$ of $2$ by $2$ traceless Hermitian matrices. For $\hat{H}$ to describe a Weyl semimetal, we tentatively require the following two properties.
\begin{enumerate}
\item[(W1)]
The set $W = \{ k \in T^3 |\ \det\hat{H}(k) = 0 \}$ is finite.

\item[(W2)]
At each point $(k^0_x, k^0_y, k^0_z) \in W$, there is an approximation
$$
\hat{H}(k_x, k_y, k_z) \sim \sum_{i = x, y, z} (k_i - k^0_i) \sigma_i.
$$
\end{enumerate}
We will call $W$ in (W1) the set of \textit{Weyl points} of $\hat{H}$. Away from the Weyl points, $\hat{H}(k)$ is invertible. Thus, as in the case of topological insulators, we have the Bloch bundle over $T^3 \backslash W$, and hence its first Chern class in the integral cohomology group
$$
c_1(E_{\hat{H}}) \in H^2(T^3 \backslash W).
$$
This is essentially the same as the (bulk) topological invariant of the Weyl semimetal $\hat{H}$ studied in \cite{M-T1,M-T2}.

Now, suppose an edge on the lattice $\Z^3$. To be concrete, let $\Z^2 \times \{ 1 \}$ be the edge. Then, as in the case of topological insulators, we can get the edge Hamiltonian, namely, a continuous family of self-adjoint operators $\{ \acute{H}^\sharp(k) \}_{k \in T^2}$ parametrized by the $2$-dimensional torus $T^2$ consisting of quasi-momenta on the boundary surface $\Z^2 \times \{ 1 \}$. A definition of the Fermi arc of $\hat{H}$ is the locus of $k \in T^2$ where $\acute{H}^\sharp(k)$ admits zero spectrum:
$$
F = \{ k \in T^2 |\ 
0 \in \sigma(\acute{H}^\sharp(k)) \},
$$
where $\sigma(\acute{H}^\sharp(k))$ is the spectrum of $\acute{H}^\sharp(k)$. Note that $\pi(W) \subset F$, where $\pi : T^3 \to T^2$ is the projection $\pi(k_x, k_y, k_z) = (k_x, k_y)$. This is because $\acute{H}^\sharp(k)$ at $k \in \pi(W)$ has $0$ as an essential spectrum. What suggested in \cite{M-T1,M-T2} as the bulk-edge correspondence is that the Fermi arc represents the homology class 
$$
\pi_*(\mathrm{PD}(c_1(E_{\hat{H}}))) \in H_1(T^3, \pi(W)),
$$
where $\mathrm{PD} : H^2(T^3 \backslash W) \to H_1(T^3, W)$ is the Poincar\'{e} duality map, and $\pi_* : H_1(T^3, W) \to H_1(T^2, \pi(W))$ induced from the projection $\pi : T^3 \to T^2$. For example, if both $W$ and $\pi(W)$ consist of distinct two points, then there is an isomorphism $H_1(T^2, \pi(W)) \cong H_1(T^2) \oplus \Z$. A basis of the summand $\Z \cong H_1(T^2, \pi(W))/H_1(T^2)$ is the homology class of a path which connects the two points in $\pi(W)$. Thus, when $F$ consists of (the images of) genuine paths, they represent a homology class in $H_1(T^2, \pi(W))$, and this homology class is expected to be $\pi_*(\mathrm{PD}(c_1(E_{\hat{H}})))$.

By definition, $F$ is just a set. Hence $F$ has no a priori parametrizations as paths. Taking this issue into account, we formulate the homological bulk-edge correspondence for Weyl semimetals as the following claim:

\begin{claim} \label{claim}
For any Weyl semimetal $\hat{H} : T^3 \to \mathrm{Herm}(\C^2)_0$ (with appropriate assumptions), there would exist a homology class
$$
\mathcal{F}(\acute{H}^\sharp) \in H_1(T^2, \pi(W))
$$
which has the following properties.
\begin{enumerate}
\item
$\mathcal{F}(\acute{H}^\sharp)$ is defined by the spectral data of the edge Hamiltonian $\{ \acute{H}^\sharp(k) \}_{k \in T^2}$ associated to $\hat{H}$.

\item
$\mathcal{F}(\acute{H}^\sharp) = \pi_*(\mathrm{PD}(c_1(E_{\hat{H}})))$.

\item
If $F \neq \emptyset$, then $F$ is recovered as the images of paths which represent $\mathcal{F}(\acute{H}^\sharp)$.

\end{enumerate}
\end{claim}

In comparison with the bulk-edge correspondence for topological insulators, $c_1(E_{\hat{H}})$ plays the role of the bulk index, and $\mathcal{F}(\acute{H}^\sharp)$ that of the edge index.

\medskip

It is rather easy to define $\mathcal{F}(\acute{H}^\sharp)$ satisfying (1) and (2) in a general setup. Actually, in this paper,  we define $\mathcal{F}(\acute{H}^\sharp)$ by the Poincar\'{e} dual of the spectral flow of $\acute{H}^\sharp|_{T^2 \backslash \pi(W)}$ regarded as a first cohomology class (see \S\S\ref{subsec:edge_index} for detail, cf.~\cite{Th2}). To the contrary, the verification of (3) for $\mathcal{F}(\acute{H}^\sharp)$, defined as above for instance, seems somehow difficult. This is because the eigenvalues of operators may merge and branch, so that the set $F$ may not have a simple description in general. Hence to find appropriate assumptions about $\hat{H}$ is the crucial part in showing Claim \ref{claim}

\medskip

Toward a general proof of Claim \ref{claim}, some non-traditional descriptions of homology groups such as in \cite{ADT,D-T} may be useful. In a recent work of Thiang \cite{Th2}, the local behaviour of zero loci of the edge Hamiltonian of a Weyl semimetal is generally analyzed. Though their global nature is not yet clarified, his result would be helpful to show Claim \ref{claim} in general.

\subsection{Main result}

As is mentioned, this paper considers a certain model of a Weyl semimetal, and then prove the homological bulk-edge correspondence. To explain our model, we introduce $\hat{H}_{\mathrm{loc}} : \R^2 \times S^1 \to \mathrm{Herm}(\C^2)_0$ by
$$
\hat{H}_{\mathrm{loc}}(a, b, \theta)
=
\left(
\begin{array}{cc}
b & a + e^{i\theta} \\
a + e^{-i\theta} & -b
\end{array}
\right).
$$
It is easy to see that $\hat{H}_{\mathrm{loc}}$ has two Weyl points $(1, 0, \pi)$ and $(-1, 0, 0)$. The associated edge Hamiltonian is a family of self-adjoint operators $\{ \acute{H}^\sharp_{\mathrm{loc}}(a, b) \}_{(a, b) \in \R^2}$, and the Fermi arc turns out to be the interval $F = \{ 0 \} \times [-1, 1]$. Taking $\hat{H}_{\mathrm{loc}}$ as the local model, we introduce our model of Weyl semimetals as follows.

\begin{assumption} \label{assumption}
We suppose that a map $\hat{H} : T^3 \to \mathrm{Herm}(\C^2)_0$  is expressed as 
$$
\hat{H}(k_x, k_y, k_z)
=
\left(
\begin{array}{cc}
b(k_x, k_y) & a(k_x, k_y) + e^{ik_z} \\
a(k_x, k_y) + e^{-ik_z} & -b(k_x, k_y)
\end{array}
\right)
$$
in terms of smooth maps $a, b : T^2 \to \R$ with the following property: 
\begin{enumerate}
\item[(a)]
The subset $a^{-1}(\{ \pm 1 \}) \cap b^{-1}(0) \subset T^2$ consists of a finite number of points.

\item[(b)]
Each $w \in \pi(W)$ admits an open neighbourhood $U_w \subset T^2$ on which we have $\det J(k_x, k_y) \neq 0$ for $(k_x, k_y) \in (U_w \backslash \{ w \}) \cap a^{-1}((-1, 1)) \cap b^{-1}(0)$, where $\det J : T^2 \to \R$ is the determinant of the Jacobian of $(a, b) : T^2 \to \R^2$.

\item[(c)]
If we restrict the domain of $b : T^2 \to \R$ to the open subset $a^{-1}((-1, 1)) \subset T^2$, then $0 \in \R$ is a regular value and the inverse image $a^{-1}((-1, 1)) \cap b^{-1}(0)$ consists of a finite number of connected components.

\item[(d)]
There exists $(k^0_x, k^0_y) \in T^2$ such that:
\begin{itemize}
\item
For any $k \in \R/2\pi \Z$, we have $(k^0_x, k), (k, k^0_y) \not\in a^{-1}(\{ \pm 1 \}) \cap b^{-1}(0)$.

\item
For any $k_y \in \R/2\pi \Z$ such that $(k^0_x, k_y) \in a^{-1}((-1, 1)) \cap b^{-1}(0)$, we have $\det J(k^0_x, k_y) \neq 0$. Also, for any $k_x \in \R/2\pi \Z$ such that $(k_x, k^0_y) \in a^{-1}((-1, 1)) \cap b^{-1}(0)$, we have $\det J(k_x, k^0_y) \neq 0$. 

\end{itemize}
\end{enumerate}
\end{assumption}

We notice that the characterization (W1) of a Weyl semimetal follows from (a), because $\pi(W) = a^{-1}(\{ \pm 1 \}) \cap b^{-1}(0)$. The characterization (W2) implies (b) when we understand the meaning of an ``approximation'' suitably (Lemma \ref{lem:characterization}). In this sense, our model is a generalization of Weyl semimetals. From (c), the Fermi arc but the projected Weyl points $F \backslash \pi(W) = a^{-1}((-1, 1)) \cap b^{-1}(0)$ is expressed as the union of a finite number of $1$-dimensional submanifolds. They are the images of smooth embeddings of intervals or circles, from which we get relative homology cycles. Then, using (b) and (d), we can verify the third property in Claim \ref{claim} to obtain our main theorem:

\begin{thm} \label{thm:main_in_introduction}
Claim \ref{claim} holds true for $\hat{H}$ in Assumption \ref{assumption}.
\end{thm}

We close the introduction by an example. It is well known that a Weyl semimetal arises in a phase transition of topological insulators (see \cite{Mur} for example). In view of the phase transitions of the QWZ model, we substitute $u = 2 + \cos k_z$ into $\hat{H}_{\mathrm{QWZ}}$ to get 
$$
\hat{H}(k_x, k_y, k_z)
= \sin k_x \sigma_x + \sin k_y \sigma_y + 
(2 + \cos k_x + \cos k_y + \cos k_z) \sigma_z.
$$
Note that one can get this from a toy model in \cite{AMV}. Under the change of coordinates $k_x \leftrightarrow k_z$ and the conjugation by an orthogonal matrix (concretely $T$ in the proof of Lemma \ref{lem:spectral_flow_of_Chern_insulator}), $\hat{H}(k_x, k_y, k_z)$ agrees with $\hat{H}_{\mathrm{loc}}(a(k_x, k_y), b(k_x, k_y), k_z)$, in which
\begin{align*}
a(k_x, k_y) &= 2 + \cos k_x + \cos k_y, &
b(k_x, k_y) &= \sin k_y.
\end{align*}
This example satisfies Assumption \ref{assumption}. The set $W$ consists of two Weyl points $(\pi/2, \pi, 0)$ and $(3\pi/2, \pi, \pi)$ on $T^3$. The Fermi arc is the interval $F = [\pi/2, 3\pi/2] \times \{ \pi \}$ connecting the projected Weyl points $(\pi/2, \pi)$ and $(3\pi/2, \pi)$ on $T^2$. It turns out that the relative homology class $\mathcal{F}(\acute{H}^\sharp) \in H_1(T^2, \pi(W)) \cong \Z^3$ is represented by the path $f : [0, 1] \to T^2$ given by $f(t) = ((2t+1)\pi/2, \pi)$. It is clear that its image recovers the Fermi arc: $\mathrm{Im}f = F$. See \S\S\ref{subsec:example} for more detail of this and other examples satisfying Assumption \ref{assumption}.

\bigskip

This paper is organized as follows: In \S\ref{sec:local_model}, we start with a summary of spectral data of the edge Hamiltonian associated to our local model $\hat{H}_{\mathrm{loc}}$. Based on this result, we also give a topological proof of the bulk-edge correspondence $c_1(\hat{H}) = -\mathrm{sf}(\acute{H}^\sharp)$ for Chern insulators here. In \S\ref{sec:Weyl_semimetal}, the main theorem is proved. For this aim, we begin with a review of relevant (co)homology groups. We next construct bases of these (co)homology groups, and define the relative homology class $\mathcal{F}(\acute{H}^\sharp)$. Then, by using the bulk-edge correspondence for Chern insulators and the (co)homology basis, we prove Theorem \ref{thm:main_in_introduction}. The proof appeals to basic techniques in differential topology. Examples are supplied at the end of this section. Finally, in \S\ref{sec:spectral_analysis}, a postponed proof about spectral data in \S\ref{sec:local_model} is given.

\medskip

Throughout the paper, some basic notions in algebraic topology (cohomology, homology and homotopy groups  \cite{Hatch}) will be assumed.


\medskip

\begin{acknowledgment}
The author's research is supported by 
JSPS KAKENHI Grant Numbers 20K03606 and JP17H06461.
\end{acknowledgment}


\section{Local model}
\label{sec:local_model}

\subsection{Summary of spectral data}
\label{subsec:summary_spectral_data}

As in the introduction, let $H_{\mathrm{loc}} : \R^2 \times S^1 \to \mathrm{Herm}(\C^2)_0$ be the following map
$$
\hat{H}_{\mathrm{loc}}(a, b, \theta)
=
\left(
\begin{array}{cc}
b & a + e^{i\theta} \\
a + e^{-i\theta} & -b
\end{array}
\right).
$$
The eigenvalues of $\hat{H}_{\mathrm{loc}}(a, b, \theta)$ is $\lambda_\pm = \pm \sqrt{b^2 + \lvert a + e^{i\theta} \rvert^2}$. Hence $\hat{H}_{\mathrm{loc}}(a, b, \theta)$ is invertible if and only if $(a, b, \theta) \neq (1, 0, \pi), (-1, 0, 0)$.

\medskip

Let $L^2(S^1, \C^2)$ be the space of $L^2$-functions on $S^1 = \R/2\pi \Z$ with values in $\C^2$, and $L^2_+(S^1, \C^2) \subset L^2(S^1, \C^2)$ the $L^2$-completion of the subspace $\bigoplus_{n > 0} \C e^{in \theta} \otimes \C^2$. Put differently, $L^2_+(S^1, \C^2)$ is the closed subspace of $L^2$-functions whose Fourier expansions involve positive modes only. We denote by $\hat{P} : L^2(S^1, \C^2) \to L^2_+(S^1, \C^2)$ the orthogonal projection. Its adjoint $\hat{P}^* : L^2_+(S^1, \C^2) \to L^2(S^1, \C^2)$ is the inclusion. For $(a, b) \in \R^2$, we define $\hat{H}^\sharp_{\mathrm{loc}}(a, b) : L^2_+(S^1, \C^2) \to L^2_+(S^1, \C^2)$ to be the compression of the multiplication operator with $\hat{H}_{\mathrm{loc}}(a, b, \cdot) : S^1 \to \mathrm{Herm}(\C^2)_0$, namely, the composition of the following operators
$$
L^2_+(S^1, \C^2) \overset{\hat{P}^*}{\longrightarrow}
L^2(S^1, \C^2) \overset{\hat{H}_{\mathrm{loc}}(a, b, \cdot)}{\longrightarrow}
L^2(S^1, \C^2) \overset{\hat{P}}{\longrightarrow}
L^2_+(S^1, \C^2).
$$
By construction, $\hat{H}^\sharp_{\mathrm{loc}}(a, b)$ is a self-adjoint bounded operator. Notice that $\hat{H}^\sharp_{\mathrm{loc}}(a, b)$ is essentially the Toeplitz operator \cite{D} associated to $\hat{H}_{\mathrm{loc}}(a, b, \cdot)$. Hence $\hat{H}^\sharp_{\mathrm{loc}}(a, b)$ is Fredholm if $(a, b) \neq (\pm 1, 0)$. Also, by using the continuity and symmetry of the model as well as the fact that essential spectrum is unchanged under compact perturbations, we can identify the essential spectrum of $\hat{H}^\sharp_{\mathrm{loc}}(a, b)$ with the union of the spectrum of the Hermitian matrices $\{ H_{\mathrm{loc}}(a, b, \theta) \}_{\theta \in S^1}$. 
$$
\sigma_{\mathrm{ess}}(\hat{H}^\sharp_{\mathrm{loc}}(a, b))
=
\bigcup_{\theta \in S^1} \sigma(\hat{H}_{\mathrm{loc}}(a, b, \theta)),
$$
where $\sigma(\hat{H}_{\mathrm{loc}}(a, b, \theta))$ is the spectrum (the set of eigenvalues) of $\hat{H}_{\mathrm{loc}}(a, b, \theta)$. Concretely, we have
\begin{align*}
\sigma_{\mathrm{ess}}(\hat{H}^\sharp_{\mathrm{loc}}(a, b))
&= 
[- \sqrt{b^2 + (\lvert a \rvert + 1)^2}, 
- \sqrt{b^2 + (\lvert a \rvert - 1)^2}] \\
& \quad
\cup
[\sqrt{b^2 + (\lvert a \rvert - 1)^2}], 
\sqrt{b^2 + (\lvert a \rvert + 1)^2}].
\end{align*}
As a result, we have $0 \in \sigma_{\mathrm{ess}}(\hat{H}^\sharp_{\mathrm{loc}}(a, b))$ if and only if $(a, b) = (\pm 1, 0)$. The other zero spectra are eigenvalues:

\begin{prop} \label{prop:spectral_data}
There are $\hat{\psi} \in L^2_+(S^1, \C^2)$ and $E \in \R$ such that $(\hat{H}^\sharp_{\mathrm{loc}}(a, b) - E) \hat{\psi} = 0$ only in the following cases.
\begin{itemize}
\item
The case that $a = 0$ and $E = \pm \sqrt{b^2 + 1}$. In this case, the space of $\hat{\psi} \in L^2_+(S^1, \C^2)$ such that $(\hat{H}^\sharp_{\mathrm{loc}}(a, b) - E) \hat{\psi} = 0$ is infinite-dimensional.

\item
The case that $\lvert a \rvert < 1$ and $E = b$. In this case, the space of $\hat{\psi} \in L^2_+(S^1, \C^2)$ such that $(\hat{H}^\sharp_{\mathrm{loc}}(a, b) - E) \hat{\psi} = 0$ is $1$-dimensional.

\end{itemize}
\end{prop}

The proof of this proposition, which is elementary,  will be given in \S\ref{sec:spectral_analysis}.

\begin{cor}
The spectrum of $\hat{H}^\sharp_{\mathrm{loc}}(a, b)$ contains zero, $0 \in \sigma(\hat{H}^\sharp_{\mathrm{loc}}(a, b))$, if and only if $(a, b) \in [-1, 1] \times \{ 0 \}$. 
\end{cor}

\subsection{Proof of bulk-edge correspondence for topological insulators}

As an application of Proposition \ref{prop:spectral_data}, we show the bulk-edge correspondence for topological insulators as mentioned in the introduction. This result will be used in the sequel. Our proof appeals to basic topology, while more analytic ones are available in \cite{Bra,Haya}, for instance.

\medskip

Let $\hat{H} : T^2 \to \mathrm{Herm}(\C^r)^*$ be a continuous map from the $2$-dimensional torus $T^2 = \R^2/2\pi \Z^2$ to the space $\mathrm{Herm}(\C^r)^*$ of invertible $r$ by $r$ Hermitian matrices. We assume that $\hat{H}(k)$ has both positive and negative eigenvalues at each $k \in T^2$. Associated to $\hat{H}$ is the Bloch vector bundle $E_{\hat{H}} \to T^2$ whose fiber at $k \in T^2$ consists of negative eigenvectors of $\hat{H}(k)$
$$
E_{\hat{H}} = \bigcup_{k \in T^2} 
\bigoplus_{\lambda < 0} \mathrm{Ker}(\hat{H}(k) - \lambda)
\subset T^2 \times \C^r.
$$
By evaluating the fundamental class $[T^2] \in H_2(T^2; \Z)$ of $T^2$ by the first Chern class $c_1(E_{\hat{H}}) \in H^2(T^2; \Z)$ of the Bloch bundle, we get the first Chern number of $E_{\hat{H}}$, which we denote by
$$
c_1(\hat{H}) = \langle c_1(E_{\hat{H}}), [T^2] \rangle \in \Z.
$$

\medskip

Let $L^2(S^1, \C^r)$ be the space of $L^2$-functions on $S^1$ with values in $\C^r$, and $L^2_+(S^1, \C^r)$ the closed subspace consisting of $L^2$-functions whose Fourier expansions involve positive modes only. For $k_x \in S^1$, we define $\hat{H}^\sharp(k_x) : L^2_+(S^1, \C^r) \to L^2_+(S^1, \C^r)$ to be the compression of the multiplication operator with $\hat{H}(k_x, \cdot) : S^1 \to \mathrm{Herm}(\C^r)^*$, namely, the composition of
$$
L^2_+(S^1, \C^r) \overset{\hat{P}^*}{\longrightarrow}
L^2(S^1, \C^r) \overset{\hat{H}(k_x, \cdot)}{\longrightarrow}
L^2(S^1, \C^r) \overset{\hat{P}}{\longrightarrow}
L^2_+(S^1, \C^r),
$$
where $\hat{P}$ is the orthogonal projection. Then $\hat{H}^\sharp(k_x)$ is a self-adjoint bounded Fredholm operator. By the assumption that $\hat{H}(k_x, k_y)$ has both positive and negative eigenvalues, $\hat{H}^\sharp(k_x)$ is neither essentially positive nor negative, in the sense of \cite{A-Si}. As a consequence, we have a continuous map $\hat{H}^\sharp : S^1 \to \mathrm{Fred}^1$, where $\mathrm{Fred}^1 = \mathrm{Fred}^1(L^2_+(S^1, \C^r))$ is the set of self-adjoint bounded Fredholm operators which are neither essentially positive nor negative, topologized by the operator norm. It is well-known that $\mathrm{Fred}^1$ is a model of the classifying space of the odd $K$-theory, and the set $[S^1, \mathrm{Fred}^1]$ of homotopy classes of continuous maps $S^1 \to \mathrm{Fred}^1$, which is the odd $K$-theory $K^1(S^1)$ of the circle, is identified with $\Z$. This identification is induced by the \textit{spectral flow} \cite{APS}. Intuitively, the spectral flow $\mathrm{sf}(A) \in \Z$ of a continuous map $A : S^1 \to \mathrm{Fred}^1$ counts how many times the eigenvalues of $A(k)$ changes from negative to positive as $k \in S^1$ goes around the circle. A rigours definition can be found in \cite{Phi}.

\begin{prop} \label{prop:bulk_edge_for_2d_insulators}
If $\hat{H} : T^2 \to \mathrm{Herm}(\C^r)^*$ is a continuous map such that $\hat{H}(k)$ admits both positive and negative eigenvalues at each $k \in T^2$, then $c_1(\hat{H}) = - \mathrm{sf}(\hat{H}^\sharp)$. 
\end{prop}

The proposition will be shown based on some lemmas below. 

\begin{lem} \label{lem:Chern_number_of_Chern_insulator}
For any integer $n \in \Z$, we define $\hat{H}_n : T^2 \to \mathrm{Herm}(\C^2)$ by
\begin{align*}
\hat{H}_n(k_x, k_y)
&= 
\sin n k_x \sigma_x + \sin k_y \sigma_y 
+ (u + \cos n k_x + \cos k_y)\sigma_z \\
&=
\left(
\begin{array}{cc}
u + \cos n k_x + \cos k_y & \sin n k_x - i \sin k_y \\
\sin n k_x + i \sin k_y & - u - \cos n k_x - \cos k_y
\end{array}
\right),
\end{align*}
where $u \in \R$ is a parameter. If $0 < u < 2$, then $\hat{H}_n(k_x, k_y)$ is invertible for each $(k_x, k_y) \in T^2$ and we have $c_1(\hat{H}_n) = n$.
\end{lem}

\begin{proof}
It is clear that $\det \hat{H}_n(k_x, k_y) \neq 0$ whenever $u \neq 0, \pm 2$. In the case that $n = 1$, the map $H_1$ is the QWZ model of Chern insulators \cite{AOP}, and its Chern number is $c_1(\hat{H}_1) = 1$ for $0 < u < 2$. The map $H_n$ factors as $\hat{H}_n = \hat{H}_1 \circ f_n$, where $f_n : T^2 \to T^2$ is given by $f_n(k_x, k_y) = (n k_x, k_y)$. Since $f_n$ carries the fundamental class $[T^2] \in H_2(T^2; \Z) \cong \Z$ of the torus to $n [T^2]$, it follows that $c_1(\hat{H}_n) = n$.
\end{proof}

Tentatively, we write $\mathrm{Herm}_p(\C^r)$ for the space of invertible Hermitian matrices which have $p$ negative eigenvalues and $r - p$ positive eigenvalues.

\begin{lem}
$\mathrm{Herm}_p(\C^r)$ is homotopy equivalent to the Grassmannian $\mathrm{Gr}_p(\C^r)$ consisting of $p$-dimensional linear subspaces in $\C^r$. Thus, if $p \neq 0, r$, then we have the following homotopy groups
\begin{align*}
\pi_0(\mathrm{Herm}_p(\C^r)) &\cong 0, &
\pi_1(\mathrm{Herm}_p(\C^r)) &\cong 0, &
\pi_2(\mathrm{Herm}_p(\C^r)) &\cong \Z.
\end{align*}
\end{lem}

\begin{proof}
Given a matrix $\hat{H} \in \mathrm{Herm}_p(\C^r)$, we have a self-adjoint involution $\eta = \hat{H}/\lvert \hat{H} \rvert$. A self-adjoint involution $\eta$ corresponds bijectively to an orthogonal projection $p = 1 - 2\eta$, and an orthogonal projection to a linear subspace in $\C^r$. Hence $\mathrm{Herm}_p(\C^r)$ contains as a subspace the Grassmannian $\mathrm{Gr}_p(\C^r)$ consisting of $p$-dimensional linear subspaces in $\C^r$. The subspace is a deformation retract, since $\hat{H}$ and $\eta$ are homotopic by the homotopy $\hat{H}/\lvert \hat{H} \vert^t$, ($t \in [0, 1]$). 
\end{proof}

\begin{lem} \label{lem:homotopy_classification}
If $p \neq 0, r$, then the set $[T^2, \mathrm{Herm}_p(\C^r)]$ of homotopy classes of continuous maps $T^2 \to \mathrm{Herm}_p(\C^r)$ is identified with $\Z$. This identification is induced by the assignment $\hat{H} \mapsto c_1(\hat{H})$ of the Chern number.
\end{lem}

\begin{proof}
Suppose that we have two continuous maps $\hat{H}, \hat{H}' : T^2 \to \mathrm{Herm}_p(\C^r)$ such that $c_1(\hat{H}) = c_1(\hat{H}')$. We shall show that $\hat{H}$ and $\hat{H}'$ are homotopic. For this aim, we notice that any continuous map $\hat{H} : T^2 \to \mathrm{Herm}_p(\C^r)$ is homotopic to $\hat{F} : T^2 \to \mathrm{Herm}_p(\C^r)$ whose restriction to the subspace $S^1 \vee S^1 = (S^1 \times \{ 0 \}) \cup (\{ 0 \} \times S^1)$ in $T^2$ is the constant map $c_p$ with values in a point $p \in \mathrm{Herm}_p(\C^r)$. The reason is as follows: Because $\mathrm{Herm}_p(\C^r)$ is simply connected, there exists a homotopy between the restriction of $\hat{H}$ to $S^1 \vee S^1$ and the constant map $c_p$. Apparently, the torus $T^2$ has the structure of a $2$-dimensional CW complex whose $1$-skeleton is $S^1 \vee S^1$. Then, applying the homotopy extension property, we can extend the homotopy on $S^1 \vee S^1$ to that on $T^2$. As a consequence, we can assume that the restrictions of $\hat{H}$ and $\hat{H}'$ to $S^1 \vee S^1$ are the constant map $c_p$. Then there is a bijective correspondence between continuous maps $T^2 \to \mathrm{Herm}_p(\C^r)$ whose restrictions to $S^1 \vee S^1$ are constant maps and continuous maps $T^2/(S^1 \vee S^1) \to \mathrm{Herm}_p(\C^r)$, where $T^2/(S^1 \vee S^1)$ is the space given by collapsing the subspace $S^1 \vee S^1$, which is homeomorphic to the sphere $S^2$. To summarize, we have the following bijections
\begin{align*}
[T^2/(S^1 \vee S^1), \mathrm{Herm}_p(\C^r)] 
&\cong [S^2, \mathrm{Herm}_p(\C^r)] \\
&\cong \pi_2(\mathrm{Herm}_p(\C^r))
\cong H_2(\mathrm{Herm}_p(\C^r); \Z) \cong \Z.
\end{align*}
The last bijection to $\Z$ is induced from the evaluation of the first Chern class of the tautological (Bloch) vector bundle on $\mathrm{Herm}_p(\C^r) \simeq \mathrm{Gr}_p(\C^r)$. Consequently,  $c_1(\hat{H}) = c_1(\hat{H}')$ implies that $\hat{H}$ and $\hat{H}'$ are homotopic.

To complete the proof, we need to show that there is $\hat{F}_n : T^2 \to \mathrm{Herm}_p(\C^r)$ such that $c_1(\hat{F}_n) = n$ for any given integer $n \in \Z$. Thanks to Lemma \ref{lem:Chern_number_of_Chern_insulator}, this is clear in the case that $p = 1$ and $r = 2$. In the general case, an example of $\hat{F}_n$ can be constructed as $\hat{F}_n(k) = \hat{H}_n(k) \oplus (-1_{\C^{p-1}}) \oplus 1_{\C^{r - p - 1}}$, where $\hat{H}_n$ is given in Lemma \ref{lem:Chern_number_of_Chern_insulator} with the parameter $0 < u < 2$, and $1_{\C^k}$ is the (constant map with its value) identity matrix on $\C^k$. 
\end{proof}

\begin{lem} \label{lem:spectral_flow_of_Chern_insulator}
For $n \in \Z$ and $1 < u < 2$, we have $\mathrm{sf}(\hat{H}^\sharp_n) = -n$, where $\hat{H}_n$ is the map given in Lemma \ref{lem:Chern_number_of_Chern_insulator}.
\end{lem}

\begin{proof}
Let $T$ be the following matrix
$$
T =
\frac{1}{\sqrt{2}}
\left(
\begin{array}{rr}
1 & 1 \\
1 & -1
\end{array}
\right).
$$
By direct computation, we see
\begin{align*}
T^* \hat{H}_n(k_x, k_y) T
&=
\left(
\begin{array}{cc}
\sin n k_x & u + \cos nk_x + e^{ik_y} \\
u + \cos n k_x + e^{-ik_y} & - \sin n k_x
\end{array}
\right) \\
&=
\hat{H}_{\mathrm{loc}}(a_n(k_x), b_n(k_x), k_y),
\end{align*}
where $a_n(k_x) = u + \cos n k_x$ and $b_n(k_x) = \sin n k_x$. It follows that $\mathrm{sf}(\hat{H}^\sharp_n) = \mathrm{sf}(\hat{H}^\sharp)$, where we set $\hat{H}(k_x, k_y) = \hat{H}_{\mathrm{loc}}(a_n(k_x), b_n(k_x), k_y)$ to suppress notations. Now, we apply Proposition \ref{prop:spectral_data}. In the case of $n = 0$, the Fredholm operator $\hat{H}^\sharp(k_x)$ turns out to be invertible for all $k_x \in S^1$, so that $\mathrm{sf}(\hat{H}^\sharp) = 0$. In the case of $n > 0$, the Fredholm operator $\hat{H}^\sharp(k_x)$ has zero eigenvalues of multiplicity one at $k_x = \pi(2\ell - 1)/n$ for $\ell = 1, 2, \ldots, n$. Around these values of $k_x$, the eigenvalues of $\hat{H}^\sharp(k_x)$ are given by $E = \sin n k_x$, so that $\mathrm{sf}(\hat{H}^\sharp) = -n$. In the case of $n < 0$, we find $\mathrm{sf}(\hat{H}^\sharp) = -n$ in the same way.
\end{proof}

We are now in the position to prove Proposition \ref{prop:bulk_edge_for_2d_insulators}:

\begin{proof}[Proof of Proposition \ref{prop:bulk_edge_for_2d_insulators}]
Because of Lemma \ref{lem:homotopy_classification}, if a map $\hat{H} : T^2 \to \mathrm{Herm}(\C^r)^*$ satisfying the hypothesis of the proposition has the Chern number $c_1(\hat{H}) = n$, then $\hat{H}$ is homotopic to $\hat{F}_n$ constructed in the proof of Lemma \ref{lem:homotopy_classification}. Then the associated families of self-adjoint Fredholm operators $\hat{H}^\sharp : S^1 \to \mathrm{Fred}^1$ and $\hat{F}^\sharp_n : S^1 \to \mathrm{Fred}^1$ are also homotopic. Since the spectral flow is a homotopy invariant, it follows that $\mathrm{sf}(\hat{H}^\sharp) = \mathrm{sf}(\hat{F}^\sharp_n)$. By Lemma \ref{lem:spectral_flow_of_Chern_insulator}, we see $\mathrm{sf}(\hat{F}_n^\sharp) = \mathrm{sf}(\hat{H}_n^\sharp) = -n$. Hence $c_1(\hat{H}) = -\mathrm{sf}(\hat{H}^\sharp)$ is established.
\end{proof}


\section{Homological bulk-edge correspondence for Weyl semimetals}
\label{sec:Weyl_semimetal}

\subsection{Relevant homology groups}

We summarize here homology groups relevant to Weyl semimetals as given in \cite{M-T1,M-T2}. All the cohomology and homology groups are those with integer coefficients.

\medskip

Let $W \subset T^3$ be a finite subset. For each $w \in W$, we choose an open ball $D(w) \subset T^3$ centered at $w \in W$ so that their closure do not intersect with each other: $\overline{D(w)} \cap \overline{D(w')} = \emptyset$ for $w \neq w'$. The ball $\overline{D(w)}$ and its boundary sphere $\partial \overline{D(w)}$ inherit orientations from $T^3$. We use this orientation to identify $H^2(\partial \overline{D(w)}) \cong \Z$. We write $i : T^3 \backslash W \to T^3$ and $i_w: \partial \overline{D(w)} \to T^3$ for the inclusions.

\begin{lem} \label{lem:Mayer_Vietoris}
There is an exact sequence of groups
$$
0 \longrightarrow
H^2(T^3) \overset{i^*}{\longrightarrow} 
H^2(T^3 \backslash W) \overset{\oplus i_w^*}{\longrightarrow}
\bigoplus_{w \in W} H^2(\partial \overline{D(w)}) 
\overset{\Sigma}{\longrightarrow} 
\Z \longrightarrow 0,
$$
where the homomorphism $\Sigma$ is given by $(q_j) \mapsto \sum_j q_j$. 
\end{lem}

\begin{proof}
As shown in \cite{M-T1,M-T2}, the exact sequence is derived from the Mayer-Vietoris exact sequence associated to the open sets $U = T^3 \backslash W$ and $V = \bigsqcup_{w \in W} D(w)$.
\end{proof}

As a result, if $W$ consists of $n$ point, then we get $H^2(T^3 \backslash W) \cong \Z^{n + 2}$ abstractly, since the exact sequence in Lemma \ref{lem:Mayer_Vietoris} admits a splitting. However, it should be noticed that there is no natural choice of a splitting. 

\medskip

Let $\pi : T^3 \to T^2$ be the projection $\pi(k_x, k_y, k_z) = (k_x, k_y)$. The image $\pi(W) \subset T^2$ of the finite subset $W \subset T^3$ is again a finite subset in $T^2$.

\begin{lem}
The projection $\pi : T^3 \to T^2$ induces the following homomorphism of exact sequences.
$$
\begin{array}{c@{}c@{}c@{}c@{}c@{}c@{}c@{}c@{}c@{}c@{}c}
0 & \ \to \ & 
H_1(T^3) & \ \to \ &
H_1(T^3, W) & \ \to \ &
H_0(W) & \ \to \ &
H_0(T^3) & \ \to \ &
0, \\
& & \downarrow & & \downarrow & & \downarrow & & \downarrow & & \\
0 & \ \to \ & 
H_1(T^2) & \ \to \ &
H_1(T^2, \pi(W)) & \ \to \ &
H_0(\pi(W)) & \ \to \ &
H_0(T^2) & \ \to \ &
0.
\end{array}
$$
\end{lem}

\begin{proof}
The exact sequences are associated to the pairs $(T^3, W)$ and $(T^2, \pi(W))$. These exact sequences are natural, and hence $\pi : T^3 \to T^2$ induces the natural homomorphisms $\pi_*$ of homology groups that constitute the homomorphism of exact sequences.
\end{proof}

\begin{lem}[Poincar\'{e} duality] \label{lem:Poincare_duality}
Let $W \subset T^3$ be a finite set. Then there is an isomorphism of groups $\mathrm{PD} : H^2(T^3 \backslash W) \to H_1(T^3, W)$. There is also an isomorphism of groups $\mathrm{PD} : H^1(T^2 \backslash \pi(W)) \to H_1(T^2, \pi(W))$.
\end{lem}

\begin{proof}
For later convenience, we account for the duality isomorphism via a manifold with boundary: The open submanifold $T^3 \backslash W \subset T^3$ is homotopy equivalent to $T^3 \backslash D$, where $D = \bigsqcup_{w \in W} D(w)$ is the disjoint union of the open balls centered at the points in $W$. Hence the homotopy axiom gives the isomorphism $H^2(T^3 \backslash W) \cong H^2(T^3 \backslash D)$. To the compact oriented $3$-dimensional manifold $T^3 \backslash D$ with its boundary $\partial (T^3 \backslash D) = \partial \overline{D} = \bigsqcup_{w \in W} \partial \overline{D(w)}$, we can apply the Poincar\'{e}(-Lefshetz) duality \cite{Hatch}
$$
H^2(T^3 \backslash D) \cong H_1(T^3 \backslash D, \partial \overline{D}).
$$
Again by the homotopy axiom, we get $H_1(T^3 \backslash D, \partial \overline{D}) \cong H_1(T^3, W)$. In summary, we have an isomorphism $H^2(T^3 \backslash W) \cong H_1(T^3, W)$. The same argument also leads to an isomorphism $H^1(T^2 \backslash W) \cong H_1(T^2, \pi(W))$. 
\end{proof}

\subsection{Construction of (co)homology basis}

This subsection is devoted to a construction of bases of relevant (co)homology groups. As before, let $W \subset T^3$ be a finite subset. Suppose that its image $\pi(W) \subset T^2$ under the projection consists of $n$ points. We express the points in $\pi(W)$ as $\overline{w}_0, \ldots, \overline{w}_{n-1}$. For $i = 0, \ldots, n-1$, let $m_i$ be the number of points in $\pi^{-1}(\overline{w}_i)$. By definition, the number of points in $W$ is $\sum_{i = 0}^{n-1} m_i$. We express the points in $\pi^{-1}(\overline{w}_i)$ as $w_i^1, \ldots, w_i^{m_i}$.

\begin{lem} \label{lem:choice}
The following holds true.
\begin{enumerate}
\item[(C1)]
There is a point $(k^0_x, k^0_y, k^0_z) \in T^3$ such that the subtori $\{ k^0_x \} \times S^1 \times S^1$, $S^1 \times \{ k^0_y \} \times S^1$ and $S^1 \times S^1 \times \{ k^0_z \}$ have no intersection with $W$.

\item[(C2)]
In the case of $n \ge 2$, there is a smooth embedding of the interval $\overline{\alpha}_i : [0, 1] \to T^2$ for $i = 1, \ldots, n-1$ such that:
\begin{itemize}
\item
$\overline{\alpha}_i(0) = \overline{w}_0$ and $\overline{\alpha}_i(1) = \overline{w}_i$ for each $i$.

\item
$\overline{\alpha}_i((0, 1]) \cap \overline{\alpha}_{i'}((0, 1]) = \emptyset$ if $i \neq i'$.

\item
The image $\overline{\alpha}_i([0, 1])$ has no intersection with the subtori $\{ k^0_x \} \times S^1$ and $S^1 \times \{ k^0_y \}$ for each $i$.

\end{itemize}

\item[(C3)]
In the case of $m_0 \ge 2$, there is a smooth embedding of intervals $\alpha_0^j : [0, 1] \to T^3$ for $j = 2, \ldots, m_0$ such that:
\begin{itemize}

\item
$\alpha_0^j(0) = w_0^1$ and $\alpha_0^j(1) = w_0^j$ for each $j$.

\item 
$\alpha_0^j((0, 1]) \cap \alpha_0^{j'}((0, 1]) = \emptyset$ if $j \neq j'$.

\item
The image $\alpha_0^j([0, 1])$ does not intersect with the subtori $\{ k^0_x \} \times S^1 \times S^1$, $S^1 \times \{ k^0_y \} \times S^1$ or $S^1 \times S^1 \times \{ k^0_z \}$ for each $j$.

\end{itemize}

\item[(C4)]
In the case of $n \ge 2$, there is a smooth embedding $\alpha_i^j : [0, 1] \to T^3$ of intervals for $i = 1, \ldots, n-1$ and $j = 1, \ldots, m_i^j$ such that:
\begin{itemize}
\item
$\pi \circ \alpha_i^j = \overline{\alpha}_i$ for each $i, j$.

\item
$\alpha_i^j(0) = w_0^1$ and $\alpha_i^j(1) = w_i^j$ for each $i, j$.

\item 
$\alpha_i^j((0, 1]) \cap \alpha_{i'}^{j'}((0, 1]) = \emptyset$ if $i \neq i'$ or $j \neq j'$, where $i$ or $i'$ can be $0$.

\item
The image $\alpha_i^j([0, 1])$ does not intersect with the subtori $\{ k^0_x \} \times S^1 \times S^1$, $S^1 \times \{ k^0_y \} \times S^1$ or $S^1 \times S^1 \times \{ k^0_z \}$ for each $i, j$.

\end{itemize}

\end{enumerate}
\end{lem}

\begin{proof}
The choice of a point $(k^0_x, k^0_y, k^0_z) \in T^3$ is obviously possible. Then $T^2$ is regarded as the square $[k^0_x, k^0_x + 2\pi] \times [k^0_y, k^0_y + 2\pi]$ with its boundary identified suitably. In this description, it is clear that the choice of the smooth paths $\overline{\alpha}_i$ is possible. In the same way, $T^3$ is regarded as the cube $[k^0_x, k^0_x + 2\pi] \times [k^0_y, k^0_y + 2\pi] \times [k^0_z, k^0_z + 2\pi]$ with its boundary identified suitably. In this description of $T^3$, we can find $\alpha_0^j$. We can also construct the smooth lifts $\alpha_i^j$ of $\overline{\alpha}_i$ for $i = 1, \ldots, n$ by choosing the $z$-component appropriately. 
\end{proof}

Under a choice of a point $(k^0_x, k^0_y, k^0_z) \in T^3$, we define
\begin{align*}
\alpha_x &: [0, 1] \to T^3, & 
\alpha_x(t) &= (2\pi t, k^0_y, k^0_z), \\
\alpha_y &: [0, 1] \to T^3, & 
\alpha_y(t) &= (k^0_x, 2\pi t, k^0_z), \\
\alpha_z &: [0, 1] \to T^3, & 
\alpha_z(t) &= (k^0_x, k^0_y, 2\pi t).
\end{align*}
We put $\overline{\alpha}_x = \pi \circ \alpha_x$ and $\overline{\alpha}_y = \pi \circ \alpha_y$, which are maps $[0, 1] \to T^2$.

\begin{lem}
Let $\alpha_x, \alpha_y, \alpha_z$ be the paths constructed from the choice of a point $(k^0_x, k^0_y, k^0_z) \in T^3$ in Lemma \ref{lem:choice} (C1).

\begin{enumerate}
\item[(a)]
Let $\overline{\alpha}_i$ be the paths chosen in Lemma \ref{lem:choice} (C2). Then the homology classes of $\overline{\alpha}_x$, $\overline{\alpha}_y$ and $\overline{\alpha}_i$ form a basis of the free abelian group $H_1(T^2, \pi(W))$.

\item[(b)]
Let $\alpha_i^j$ be the paths chosen in Lemma \ref{lem:choice} (C3) and (C4). Then the homology classes of $\alpha_x$, $\alpha_y$, $\alpha_z$ and $\alpha_i^j$ form a basis of the free abelian group $H_1(T^3, W)$

\end{enumerate}
\end{lem}

\begin{proof}
We consider (a) only, since (b) can be shown in the same way. The homology exact sequence for the pair $(T^3, W)$ reads
$$
\underbrace{H_1(\pi(W))}_0 \to
\underbrace{H_1(T^2)}_{\Z^2} \overset{j_*}{\to}
H_1(T^2, \pi(W)) \overset{\partial_*}{\to}
\underbrace{H_0(\pi(W))}_{\Z^n} \overset{j_*}{\to}
\underbrace{H_0(T^2)}_{\Z},
$$
in which $j_* : H_0(\pi(W)) \to H_0(T^2)$ is the homomorphism $(q_i) \mapsto \sum_i q_i$. Hence $H_1(T^2, \pi(W)) \cong \Z^{n + 1}$. The homology classes of the paths $\overline{\alpha}_x$ and $\overline{\alpha}_y$ form a basis of $H_1(T^2) \cong \Z^2$. They inject into $H_1(T^2, \pi(W)$ under the injection $j_* : H_1(T^2) \to H_1(T^2, \pi(W))$. We take them as a part of a basis of $H_1(T^2, \pi(W))$. To see that the homology classes of $\overline{\alpha}_1, \ldots, \overline{\alpha}_{n-1}$ constitute the remaining part of a basis of $H_1(T^2, \pi(W))$, it suffices to see that their image under the boundary map $\partial_* : H_1(T^2, \pi(W)) \to H_0(\pi(W))$ form a basis of the kernel of $j_* : H_0(\pi(W)) \to H_0(T^2)$. This certainly holds true, because if we let a continuous path $\alpha : [0, 1] \to T^2$ such that $\alpha(0), \alpha(1) \in \pi(W)$ represent a homology class in $H_1(T^2, \pi(W))$, then its image under the boundary homomorphism is $\partial_*[\alpha] = \alpha(1) - \alpha(0)$. 
\end{proof}

\begin{dfn}
Under choices in Lemma \ref{lem:choice}, we define a basis $\{ \omega_x, \omega_y, \omega_z, \omega_i^j \}$ of the group $H^2(T^3 \backslash W)$ to be 
\begin{align*}
\omega_x &= \mathrm{PD}^{-1}([\alpha_x]), &
\omega_y &= \mathrm{PD}^{-1}([\alpha_y]), &
\omega_z &= \mathrm{PD}^{-1}([\alpha_z]), &
\omega_i^j &= \mathrm{PD}^{-1}([\alpha_i^j]).
\end{align*}
\end{dfn}

\begin{prop} \label{prop:cohomology_basis}
Let $\{ \omega_x, \omega_y, \omega_z, \omega_i^j \}$ be the basis of $H^2(T^3 \backslash W)$ associated to the choices in Lemma \ref{lem:choice}. Any $\omega \in H^2(T^3 \backslash W)$ is expressed as 
$$
\omega = \sum_{i = x, y, z} q_i \omega_i
+ \sum_{j = 2}^{m_0} q_0^j \omega_0^j
+ \sum_{i = 1}^{n-1} \sum_{j = 1}^{m_i} q_i^j \omega_i^j,
$$
 where the integers $q_x$, $q_y$ and $q_z$ are given by
\begin{align*}
q_x &= \int_{\{ k^0_x \} \times S^1 \times S^1} \omega, &
q_y &= -\int_{\ S^1 \times \{ k^0_y \} \times S^1} \omega, &
q_z &= \int_{\ S^1 \times S^1 \times \{ k^0_z \}} \omega,
\end{align*}
and the integers $q_i^j$ by
$$
q_i^j = - \int_{\partial \overline{D(w_i^j)}} \omega.
$$
\end{prop}

In the lemma above, the evaluation of the fundamental class $[M] \in H^2(M)$ of a closed oriented $2$-dimensional submanifold $M \subset T^3$ by the restriction of $\omega$ is denoted with $\int_M \omega = \langle \omega|_M, [M] \rangle \in \Z$. We take the orientations on the $2$-dimensional subtori to be the standard ones on them, while the orientations on the spheres $\partial \overline{D(w)}$ are induced from $\partial \overline{D(w)} \subset \overline{D(w)} \subset T^3$.

\begin{proof}
To see the lemma, we review the construction of the Poincar\'{e} dual $\omega_j = \mathrm{PD}^{-1}([\alpha_i^j])$ of $[\alpha_i^j] \in H_1(T^3, W)$. In view of the proof of Lemma \ref{lem:Poincare_duality}, we can consider $H^2(T^3 \backslash D)$ instead of $H^2(T^3 \backslash W)$. Furthermore, since the cohomology group is torsion free, we can regard cohomology classes in $H^2(T^3 \backslash D)$ as represented by differential forms. Then we can apply the construction of the Poincar\'{e} dual in \cite{B-T}. Choosing the open disks $D(w)$ smaller if necessary, we truncate and rescale the domain of the smooth path $\alpha_i^j : [0, 1] \to T^3 \backslash W$ for each $i, j$ to get a smooth path $\lvert \alpha_i^j \rvert : [0, 1] \to T^3 \backslash D$ such that $\lvert \alpha_i^j \rvert(0) \in \partial \overline{D(w_0^1)}$, $\lvert \alpha_i^j \rvert(1) \in \partial \overline{D(w_i^j)}$ and $\mathrm{Im}(\lvert \alpha_i^j \rvert) \subset \mathrm{Im}(\alpha_i^j)$. There is then an open tubular neighbourhood $N_i^j$ of the submanifold $\lvert \alpha_i^j \rvert([0, 1])$ in $T^3 \backslash D$ whose intersection with the boundary $\partial (T^3 \backslash D)$ consists of open $2$-disks on $\partial \overline{D(w_0)}$ and $\partial \overline{D(w_i^j)}$. There is a diffeomorphism $\phi_i^j : N_i^j \to [0, 1] \times D^2$, where $D^2$ is an open disk. Let $\rho_{D^2}$ be a closed $2$-form on $D^2$ with compact support such that $\int_{D^2} \rho_{D^2} = 1$. We use the same symbol $\rho_{D^2}$ to mean its pull-back under the projection $[0, 1] \times D^2 \to D^2$. The extension of the $2$-form $(\phi_i^j)^*\rho_{D^2}$ to $T^3 \backslash D$ by zero represents $\mathrm{PD}^{-1}([\alpha_i^j])$ in $H^2(T^3 \backslash D)$. The duals $\omega_x$, $\omega_y$ and $\omega_z$ admit similar representatives by $2$-forms whose supports are in tubular neighbourhoods of the embedded circles $\alpha_x([0, 1])$, $\alpha_y([0, 1])$ and $\alpha_z([0, 1])$.

Because we can choose $N_i^j$ to be disjoint to each other, we have $\int_{\partial \overline{D(w_i^j)}} \omega_{i'}^{j'} = 0$ whenever $i \neq i'$ or $j \neq j'$. By construction, we also have $\int_{\partial \overline{D(w_i^j)}} \omega_i^j = -1$, where the sign is due to the difference of the orientations on $\partial (T^3 \backslash D) = \partial \overline{D}$ induced from $T^3$. Since the tubular neighbourhoods of $\alpha_x, \alpha_y, \alpha_z$ can be chosen so as to be disjoint to $D$, we have $\int_{\partial \overline{D(w_i^j)}} \omega_k = 0$ for $k = x, y, z$. Thus, in view of Lemma \ref{lem:Mayer_Vietoris}, the coefficients $q_i^j$ are given by $q_i^j = -\int_{\partial \overline{D(w_i^j)}} \omega$. Since the images of the paths $\alpha_i^j$ have no intersection with the subtori $\{ k^0_x \} \times S^1 \times S^1$, $S^1 \times \{ k^0_y \} \times S^1$ or $S^1 \times S^1 \times \{ k^0_z \}$, the integrations of $\omega_i^j$ on these subtori are trivial. Note that $\omega_x$, $\omega_y$ and $\omega_z$ can be regarded as the injective image of a basis of $H^2(T^3) \cong \Z^3$. Then they can be represented by $2$-forms $dk_y dk_z$, $- dk_x dk_z$ and $dk_x dk_z$, respectively. This description leads to the expressions of $q_x$, $q_y$ and $q_z$ as stated.
\end{proof}

\begin{cor} \label{cor:cohomology_basis}
Under choices in Lemma \ref{lem:choice}, we have the following expression of the push-forward of the Poincar\'{e} dual of any $\omega \in H^2(T^3 \backslash W)$
\begin{align*}
\pi_*\mathrm{PD}(\omega)
&= 
\bigg( \int_{\{ k^0_x \} \times S^1 \times S^1} \omega \bigg) 
[\overline{\alpha}_x]
- 
\bigg( \int_{\ S^1 \times \{ k^0_y \} \times S^1} \omega \bigg) 
[\overline{\alpha}_y]
\\
&\quad 
- \sum_{i = 1}^{n-1} 
\bigg( \sum_{j = 1}^{m_i} \int_{\partial \overline{D(w_i^j)}} \omega \bigg) 
[\overline{\alpha}_i].
\end{align*}
\end{cor}

\subsection{Definition of the edge index}
\label{subsec:edge_index}

Let $\hat{H} : T^3 \to \mathrm{Herm}(\C^r)$ be a continuous map, where $r \ge 1$. Define $\hat{H}^\sharp(k_x, k_y) : L^2_+(S^1, \C^r) \to L^2_+(S^1, \C^r)$ to be the self-adjoint bounded operator obtained as the compression of the multiplication with $\hat{H}(k_x, k_y, \cdot)$ by the projection $L^2(S^1, \C^r) \to L^2_+(S^1, \C^r)$. The Fourier transform identifies $L^2(S^1, \C^r)$ with the space $L^2(\Z, \C^r)$ of square summarable functions on $\Z$ with values in $\C^r$
$$
L^2(\Z, \C^r) =
\bigg\{ \psi = (\psi(n))_{n \in \Z} \bigg|\
\psi(n) \in \C^r, \sum_{n \in \Z} \lVert \psi(n) \rVert^2 < + \infty \bigg\},
$$
and the subspace $L^2_+(S^1, \C^r) \subset L^2(S^1, \C^r)$ with $L^2(\mathbb{N}, \C^r) \subset L^2(\Z, \C^r)$. The operator $\acute{H}^\sharp(k_x, k_y) : L^2(\mathbb{N}, \C^r) \to L^2(\mathbb{N}, \C^r)$ mentioned in \S\ref{sec:introduction} is the one arising as the Fourier transform of $\hat{H}^\sharp(k_x, k_y)$. Since $\acute{H}^\sharp(k_x, k_y)$ and $\hat{H}^\sharp(k_x, k_y)$ have the same spectral data, we will work with $\hat{H}^\sharp(k_x, k_y)$.

\medskip

Let $W = \{ k \in T^3 |\ \det\hat{H}(k) = 0 \}$ be the subset on that $\hat{H}$ is not invertible, and $\pi : T^3 \to T^2$ the projection $\pi(k_x, k_y, k_z) = (k_x, k_y)$ as before. If $(k_x, k_y) \in T^2 \backslash \pi(W)$, then $\hat{H}(k_x, k_y, k_z)$ is invertible for $k_z \in S^1$, so that $\hat{H}^\sharp(k_x, k_y)$ is Fredholm. Assume that $\hat{H}^\sharp(k_x, k_y)$ has both positive and negative essential spectrum. In this case, we have a map $\hat{H}^\sharp : T^2 \backslash \pi(W) \to \mathrm{Fred}^1(L^2_+(S^1, \C^r))$, which is continuous.

\medskip

In general, if $X$ is a path connected space and $A : X \to \mathrm{Fred}^1$ is a continuous map, then we can construct a homology class $\mathrm{Sf}(A) \in H^1(X)$ by using spectral flows: Given a closed path $\gamma : S^1 \to X$, the spectral flow $\mathrm{sf}(A \circ \gamma) \in \Z$ makes sense. We can then define a map $\mathrm{Sf}(A) : \pi_1(X) \to \Z$ to be $(\mathrm{Sf}(A))([\gamma]) = \mathrm{sf}(A \circ \gamma)$. It turns out that $\mathrm{Sf}(A)$ is a homomorphism. The universal coefficient theorem and the Hurewicz theorem  \cite{Hatch} lead to the isomorphisms
$$
H^1(X) \cong \mathrm{Hom}(H_1(X); \Z)
\cong \mathrm{Hom}(\pi_1(X); \Z).
$$
Therefore we get $\mathrm{Sf}(A) \in H^1(X)$ for any continuous map $A : X \to \mathrm{Fred}^1$. Applying this general construction to $\hat{H}^\sharp|_{T^2 \backslash \pi(W)} : T^2 \backslash \pi(W) \to \mathrm{Fred}^1$, we define our edge index.

\begin{dfn} \label{dfn:edge_index}
Let $\hat{H} : T^3 \to \mathrm{Herm}(\C^r)$ be a continuous map such that $W = \{ k \in T^3 |\ \det\hat{H}(k) = 0 \}$ is a finite set. When $\hat{H}^\sharp(k_x, k_y)$ is neither essentially positive nor negative for each $(k_x, k_y) \in T^2 \backslash \pi(W)$, we define $\mathcal{F}(\hat{H}^\sharp) \in H_1(T^2, \pi(W))$ to be the Poincar\'{e} dual of $\mathrm{Sf}(H^\sharp|_{T^2 \backslash \pi(W)}) \in H^1(T^2 \backslash \pi(W))$ multiplied by $-1$,
$$
\mathcal{F}(\hat{H}^\sharp) = 
- \mathrm{PD}(\mathrm{Sf}(\hat{H}^\sharp|_{T^2 \backslash \pi(W)}))
\in H_1(T^2, \pi(W)).
$$
\end{dfn}

\begin{lem} \label{lem:spectral_flow_class_in_basis_expression}
Let $\hat{H} : T^3 \to \mathrm{Herm}(\C^r)$ be as in Definition \ref{dfn:edge_index}. Under choices in Lemma \ref{lem:choice}, we have
$$
\mathcal{F}(\hat{H}^\sharp) 
= -\mathrm{sf}(\hat{H}^\sharp|_{\{ k^0_x \} \times S^1}) [\overline{\alpha}_x]
+ \mathrm{sf}(\hat{H}^\sharp|_{S^1 \times \{ k^0_y \}}) [\overline{\alpha_y}]
+ \sum_{i = 1}^{n-1} 
\mathrm{sf}(\hat{H}^\sharp|_{\partial \overline{D}_i}) [\overline{\alpha}_i],
$$
where $D_1, \ldots, D_n \subset T^2$ are open disks such that $\overline{w}_i \in D_i$ and $\overline{D}_i \cap \overline{D}_j = \emptyset$ for $i \neq j$.
\end{lem}

In Lemma \ref{lem:spectral_flow_class_in_basis_expression}, we give $S^1 \times \{ k^0_y \} = \{ k^0_x \} \times S^1 = S^1$ the standard parametrization, and regard $\hat{H}^\sharp|_{S^1 \times \{ k^0_y \}}$ and $\hat{H}^\sharp|_{\{ k^0_x \} \times S^1}$ as continuous maps $S^1 \to \mathrm{Fred}^1$. Similarly, we regard $\hat{H}^\sharp|_{\partial \overline{D}_i} : S^1 \to \mathrm{Fred}^1$, taking the ``anticlockwise'' parametrization on the boundary circle $\partial \overline{D}_i = S^1$, namely, the orientation on $\overline{D}_i$ compatible with the parametrization agrees with that induced from $T^2$.

\begin{proof}
In the same manner as in the proof of Proposition \ref{prop:cohomology_basis}, we have
\begin{align*}
\mathrm{PD}(\omega)
&= \bigg( \int_{\{ k^0_x \} \times S^1} \omega \bigg)
[\overline{\alpha}_x]
- \bigg( \int_{S^1 \times \{ k^0_y \}} \omega \bigg)
[\overline{\alpha_y}]
-
\sum_{i = 1}^{n-1}
\bigg(
\int_{\partial \overline{D}_i} \omega
\bigg)
[\overline{\alpha}_i]
\end{align*}
for any $\omega \in H^1(T^2 \backslash \pi(W))$. If $\omega = -\mathrm{Sf}(\hat{H}^\sharp|_{T^2 \backslash \pi(W)})$, then
$$
\int_{\{ k^0_x \} \times S^1} \omega
= 
-(\mathrm{Sf}(\hat{H}^\sharp|_{T^2 \backslash \pi(W)}))(\{ k^0_x \} \times S^1)
= -\mathrm{sf}(\hat{H}^\sharp|_{\{ k^0_x \} \times S^1}).
$$
The same expressions hold for other coefficients of the base elements.
\end{proof}

\begin{prop} \label{prop:push_forward_of_bulk_index}
The homology class $\mathcal{F}(\hat{H}^\sharp)$ in Definition \ref{dfn:edge_index} satisfies 
$$
\pi_*\mathrm{PD}(c_1(E_{\hat{H}})) = \mathcal{F}(\hat{H}^\sharp).
$$
\end{prop}

\begin{proof}
By Corollary \ref{cor:cohomology_basis}, we have
\begin{align*}
\pi_*\mathrm{PD}(c_1(E_{\hat{H}}))
&= c_1(\hat{H}|_{\{ k^0_x \} \times S^1 \times S^1}) 
[\overline{\alpha}_x]
- c_1(\hat{H}|_{S^1 \times \{ k^0_y \} \times S^1}) 
[\overline{\alpha_y}] \\
&\quad -
\sum_{i = 1}^{n-1}
\bigg(
\sum_{j = 1}^{m_i} \int_{\partial \overline{D(w_i^j)}} c_1(E_{\hat{H}})
\bigg)
[\overline{\alpha}_i].
\end{align*}
For the moment, we choose and fix $i = 1, \ldots, n-1$. For each $j = 1, \dots, m_i$, we can choose the disk $D(w_i^j)$ so as to be $\overline{D(w_i^j)} \subset \pi^{-1}(D_i)$. The inverse image $\pi^{-1}(\overline{D}_i) \subset T^3$ of the closed disk $\overline{D}_i \subset T^2$ is homeomorphic to the solid torus $\overline{D}_i \times S^1$. Regarding $\pi^{-1}(\overline{D}_i) \backslash \sqcup_{j} D(w_i^j)$ as a homology $2$-chain, we get
$$
\sum_{j = 1}^{m_i} \int_{\partial \overline{D(w_i^j)}} c_1(E_{\hat{H}})
= \int_{\partial \overline{D}_i \times S^1} c_1(E_{\hat{H}})
= c_1(\hat{H}|_{\partial \overline{D}_i \times S^1}).
$$
Now, Proposition \ref{prop:bulk_edge_for_2d_insulators} and Lemma \ref{lem:spectral_flow_class_in_basis_expression} complete the proof.
\end{proof}

\subsection{Proof of bulk-edge correspondence for Weyl semimetals}

We here prove our main theorem (Theorem \ref{thm:main_in_introduction}). Recall that, in Assumption \ref{assumption}, a map $\hat{H} : T^3 \to \mathrm{Herm}(\C^2)_0$ is supposed to be expressed as 
$$
\hat{H}(k_x, k_y, k_z)
=
\left(
\begin{array}{cc}
b(k_x, k_y) & a(k_x, k_y) + e^{ik_z} \\
a(k_x, k_y) + e^{-ik_z} & -b(k_x, k_y)
\end{array}
\right)
$$
in terms of smooth maps $a, b : T^2 \to \R$ with the following properties: 
\begin{enumerate}
\item[(a)]
The subset $a^{-1}(\{ \pm 1 \}) \cap b^{-1}(0) \subset T^2$ consists of a finite number of points.

\item[(b)]
Each $w \in \pi(W)$ admits an open neighbourhood $U_w \subset T^2$ on which we have $\det J(k_x, k_y) \neq 0$ for $(k_x, k_y) \in (U_w \backslash \{ w \}) \cap a^{-1}((-1, 1)) \cap b^{-1}(0)$, where $\det J : T^2 \to \R$ is the determinant of the Jacobian of $(a, b) : T^2 \to \R^2$.

\item[(c)]
If we restrict the domain of $b : T^2 \to \R$ to the open subset $a^{-1}((-1, 1)) \subset T^2$, then $0 \in \R$ is a regular value and the inverse image $a^{-1}((-1, 1)) \cap b^{-1}(0)$ consists of a finite number of connected components.

\item[(d)]
There exists $(k^0_x, k^0_y) \in T^2$ such that:
\begin{itemize}
\item
For any $k \in \R/2\pi \Z$, we have $(k^0_x, k), (k, k^0_y) \not\in a^{-1}(\{ \pm 1 \}) \cap b^{-1}(0)$.

\item
For any $k_y \in \R/2\pi \Z$ such that $(k^0_x, k_y) \in a^{-1}((-1, 1)) \cap b^{-1}(0)$, we have $\det J(k^0_x, k_y) \neq 0$. Also, for any $k_x \in \R/2\pi \Z$ such that $(k_x, k^0_y) \in a^{-1}((-1, 1)) \cap b^{-1}(0)$, we have $\det J(k_x, k^0_y) \neq 0$.

\end{itemize}
\end{enumerate}

Theorem \ref{thm:main_in_introduction} is equivalent to:

\begin{thm} \label{thm:main}
Let $\hat{H} : T^3 \to \mathrm{Herm}(\C^2)_0$ be as in Assumption \ref{assumption}. Then there exists a homology class $\mathcal{F}(\hat{H}^\sharp) \in H_1(T^2, \pi(W))$ such that:
\begin{enumerate}
\item
$\mathcal{F}(\hat{H}^\sharp)$ is defined by the spectral data of the edge Hamiltonian $\{ \hat{H}^\sharp(k) \}_{k \in T^2}$ associated to $\hat{H}$.

\item
$\mathcal{F}(\hat{H}^\sharp) = \pi_*(\mathrm{PD}(c_1(E_{\hat{H}})))$.

\item
If $F \neq \emptyset$, then $F$ is recovered as the images of paths which represent $\mathcal{F}(\acute{H}^\sharp)$.

\end{enumerate}
\end{thm}

The homology class $\mathcal{F}(\hat{H}^\sharp)$ is of course the one as given in Definition \ref{dfn:edge_index}. The proof of the theorem will be given, after some lemmas.

\medskip

By definition, $\hat{H}(k_x, k_y, k_z)$ is expressed as
$$
\hat{H}(k_x, k_y, k_z) = \hat{H}_{\mathrm{loc}}(a(k_x, k_y), b(k_x, k_y), k_z)
$$
in terms of the local model $\hat{H}_{\mathrm{loc}}$. Hence the set of Weyl points $W \subset T^3$, its image under the projection $\pi : T^3 \to T^2$ and the Fermi arc $F \subset T^2$ are given by
\begin{align*}
W &= \{ (k_x, k_y, 0) |\ a(k_x, k_y) = -1, b(k_x, k_y) = 0 \} \\
&\quad \cup \{ (k_x, k_y, \pi) |\ a(k_x, k_y) = 1, b(k_x, k_y) = 0 \}, \\
\pi(W) &= \{ (k_x, k_y) |\ a(k_x, k_y) = \pm 1, b(k_x, k_y) = 0 \}
= a^{-1}(\{ \pm 1 \}) \cap b^{-1}(0), \\
F &= \{ (k_x, k_y) |\ \lvert a(k_x, k_y) \rvert \le 1, b(k_x, k_y) = 0 \}
= a^{-1}([-1, 1]) \cap b^{-1}(0).
\end{align*}
By Assumption \ref{assumption} (a), the sets $W$ and $\pi(W)$ consist of finite number of points.

\begin{lem} \label{lem:fermi_arc_as_submanifold}
If $a^{-1}((-1, 1)) \cap b^{-1}(0) \neq \emptyset$, then there are a finite number of smooth embeddings $f_i : (0, 1) \to T^2$, ($i = 1, \ldots, \ell$) and $f_i^c : S^1 \to T^2$, ($i = 1, \ldots, \ell^c$) such that
$$
a^{-1}((-1, 1)) \cap b^{-1}(0) = 
f_1((0, 1)) \sqcup \cdots f_\ell((0, 1)) \sqcup
f_1^c(S^1) \sqcup \cdots f_{\ell^c}^c(S^1).
$$
\end{lem}

\begin{proof}
By Assumption \ref{assumption} (c), the smooth map $b : a^{-1}((-1, 1)) \to \R$ has $0 \in \R$ as a regular value. By the so-called preimage theorem, $a^{-1}((-1, 1) \cap b^{-1}(0) \subset a^{-1}((-1, 1))$ is a $1$-dimensional submanifold in the open submanifold $a^{-1}((-1, 1)) \subset T^2$. This is supposed to have a finite number of connected components in Assumption \ref{assumption} (c). By the classification of connected $1$-dimensional manifolds, each of the connected components is either an embedded open interval or a circle $S^1$. 
\end{proof}

\begin{lem} \label{lem:extension}
If $a^{-1}((-1, 1)) \cap b^{-1}(0) \neq \emptyset$ and $\ell > 0$ in Lemma \ref{lem:fermi_arc_as_submanifold}, then, for $i = 1, \ldots, \ell$, the smooth embedding $f_i : (0, 1) \to T^2$ extends to a continuous map $f_i : [0, 1] \to T^2$ so that $f_i(0), f_i(1) \in \pi(W)$.
\end{lem}

\begin{proof}
We take a metric on $T^2$ compatible with the topology. Then $T^2$ is complete, because $T^2$ is compact. If we choose a sequence $t_n \in (0, 1)$ such that $\lim_{n \to \infty} t_n = 0$, then the embedded image $f_i(t_n) \in T^2$ forms a Cauchy sequence. Thus, there exists a limit point $\lim_{n \to \infty}f_i(t_n) \in T^2$, which is unique and independent of the sequence. We define $f_i(0) \in T^2$ to be this limit point. In the same way, we can define $f_i(1) \in T^2$. By construction, the extension $f_i : [0, 1] \to T^2$ is continuous. The continuity of $b : T^2 \to \R$ leads to $b(f_i(0)) = b(f_i(1)) = 0$. Because $f_i((0, 1)) \subset a^{-1}((-1, 1)) \cap b^{-1}(0)$, we see $f_i(0), f_i(1) \in \overline{a^{-1}((-1, 1)) \cap b^{-1}(0)}$. Recall that $a^{-1}((-1, 1)) \cap b^{-1}(0)$ is the disjoint union of connected $1$-dimensional submanifolds. From this, it follows that $f(0), f(1) \not\in a^{-1}((-1, 1)) \cap b^{-1}(0)$. We have
\begin{align*}
&\overline{ a^{-1}((-1, 1)) \cap b^{-1}(0) } \backslash 
a^{-1}((-1, 1)) \cap b^{-1}(0) \\
&\quad =
\overline{ (a, b)^{-1}((-1, 1) \times \{ 0 \}) } \backslash
(a, b)^{-1}((-1, 1) \times \{ 0 \}) \\
&\quad \subset
(a, b)^{-1}( \overline{(-1, 1) \times \{ 0 \}} ) \backslash
(a, b)^{-1}((-1, 1) \times \{ 0 \}) = \pi(W).
\end{align*}
Hence $f(0), f(1) \in \pi(W)$.
\end{proof}

\begin{lem} \label{lem:function_associated_to_Fermi_arc}
Let $f : (0, 1) \to T^2$ be one of the smooth embeddings $f_1, \ldots, f_\ell$ in Lemma \ref{lem:fermi_arc_as_submanifold}. Expressing $f(t) = (f_x(t), f_y(t)) \in T^2$, we put
$$
c_f(t) = 
\frac{
\displaystyle{
-\frac{\partial b}{\partial k_x}(f(t))
\frac{\partial f_y}{\partial t}(t)
+ 
\frac{\partial b}{\partial k_y}(f(t))
\frac{\partial f_x}{\partial t}(t)
}
}
{
\displaystyle{
\bigg(\frac{\partial b}{\partial k_x}(f(t))\bigg)^2 +
\bigg(\frac{\partial b}{\partial k_y}(f(t))\bigg)^2
}
}.
$$
Then the following holds true.
\begin{itemize}
\item[(a)]
$c_f(t)$ gives rise to a well-defined continuous map $c_f : (0, 1) \to \R$ such that $c_f(t) \neq 0$ for all $t$. As a result, the following sign is also well-defined
$$
\epsilon(f) = \mathrm{sign}(c_f(t))
= \frac{c_f(t)}{\lvert c_f(t) \rvert}.
$$

\item[(b)]
We have the following formula for all $t$
$$
\det J(f(t)) = c_f(t) \frac{d (a \circ f)}{d t}(t).
$$

\end{itemize}
If $f: S^1 \to T^2$ is one of the smooth embeddings $f_1^c, \ldots, f_{\ell^c}^c$ in Lemma \ref{lem:fermi_arc_as_submanifold}, then we can similarly define $c_f : S^1 \to \R$ and $\epsilon(f) \in \{ \pm 1 \}$ with the above properties.
\end{lem}

\begin{proof}
To suppress notation, let us introduce vectors $v, w \in \R^2$ as follows
\begin{align*}
v &=
\bigg(
\frac{\partial b}{\partial k_x}(f(t)),
\frac{\partial b}{\partial k_y}(f(t))
\bigg), &
w &=
\bigg(
\frac{d f_x}{dt}(t), 
\frac{d f_y}{dt}(t)
\bigg).
\end{align*}
Recall that $f : (0, 1) \to T^2$ corresponds to one of the connected components of the submanifold $a^{-1}((-1, 1)) \cap b^{-1}(0)$, where $0$ is a regular value of $b : T^2 \to \R$ restricted to $a^{-1}((-1, 1)) \subset T^2$. Thus, the differential of $b$ does not vanish along $f$, and hence $v \neq 0$. As a result, $c_f(t)$ is well-defined. The continuity of $c_f : (0, 1) \to \R$ is clear. To see that $c_f(t) \neq 0$, recall $b(f(t)) = 0$ by the definition of $f$. Then, 
$$
0 = \frac{d}{dt} b(f(t))
= \frac{\partial b}{\partial k_x}(f(t)) \frac{d f_x}{d t}(t)
+ \frac{\partial b}{\partial k_y}(f(t)) \frac{d f_y}{d t}(t).
$$
This means that $v$ and $w$ are orthogonal, or equivalently, $v$ is parallel to $w$ rotated by $\pi/2$. As is seen, $v \neq 0$. We also have $w \neq 0$, because $f$ is an embedding. As a result, there uniquely exists a real number $c \neq 0$ such that
\begin{align*}
c \frac{\partial b}{\partial k_x}(f(t)) &= - \frac{df_y}{dt}(t), &
c \frac{\partial b}{\partial k_y}(f(t)) &= \frac{df_x}{dt}(t),
\end{align*}
when $t$ is fixed. It is easy to see that $c = c_f(t)$, and (a) is proved. We omit the verification of (b), which is straightforward. Finally, the argument so far is local, and hence applicable to $f = f_1^c, \ldots, f^c_{\ell^c}$ as well.
\end{proof}

\begin{lem} \label{lem:spectral_flow_as_intersection_number}
Suppose $a^{-1}((-1, 1)) \cap b^{-1}(0) \neq \emptyset$. Let $\gamma : S^1 \to T^2 \backslash \pi(W)$ be an embedding of a circle such that:
\begin{itemize}
\item
the embedding $\gamma$ intersects with each of the submanifolds corresponding to $f_1, \ldots, f_\ell, f^c_1, \ldots, f^c_{\ell^c}$ transversally;

\item
$\det J(p) \neq 0$ at each point $p$ that $\gamma$ intersects with $f_1, \ldots, f_\ell, f^c_1, \ldots, f^c_{\ell^c}$.

\end{itemize}
Then we have
$$
\mathrm{sf}(\hat{H}^\sharp \circ \gamma)
= \sum_{j = 1}^\ell \epsilon(f_j) I(f_j, \gamma)
+ \sum_{j = 1}^{\ell^c} \epsilon(f^c_j) I(f^c_j, \gamma),
$$
where $\epsilon(f_j)$ and $\epsilon(f^c_j)$ are given in Lemma \ref{lem:function_associated_to_Fermi_arc}, and 
\begin{align*}
I(f_j, \gamma) &= \int_{\gamma} \mathrm{PD}^{-1}(f_j), &
I(f^c_j, \gamma) &= \int_{\gamma} \mathrm{PD}^{-1}(f^c_j).
\end{align*}
\end{lem}

\begin{proof}
In view of the construction of the Poincar\'{e} dual $\mathrm{PD}^{-1}(f_j)$ as reviewed in the proof of Lemma \ref{lem:Poincare_duality}, one sees that the integer $I(f_j, \gamma)$ agrees with the so-called intersection number, namely, the number of positive intersection points of $f_j$ and $\gamma$ minus that of negative intersection points. Here an intersection point $p \in T^2$ of $f_j$ and $\gamma$ is positive (resp.\ negative) if the pair of tangent vectors $(df_j/dt, d\gamma/dt)$ form a basis of the tangent space at $p \in T^2$ compatible (resp.\ incompatible) with the orientation on $T^2$. By assumption, the determinant of the Jacobian $\det J$ of $(a, b) : T^2 \to \R^2$ does not vanish at each intersection point $p$ of $f_j$ and $\gamma$. Thus, near this point, $(a, b)$ is a diffeomorphism. Because $f_j$ corresponds to a connected component of the submanifold $a^{-1}((-1, 1)) \cap b^{-1}(0)$, the path $(a, b) \circ f = (a \circ f, 0)$ lies on the interval $(-1, 1) \times \{ 0 \} \subset \R^2$. With this interval, the path $(a, b) \circ \gamma$ intersects transversally. For $s \in S^1$ such that $-1 < a(\gamma(s)) < 1$, the unique eigenvalue of $(\hat{H}^\sharp \circ \gamma)(s)$ with finite multiplicity is $E = b(\gamma(s))$ of multiplicity one. Therefore an intersection point $p = \gamma(s)$ of $f_j$ and $\gamma$ corresponds bijectively to the zero eigenvalue of $(\hat{H}^\sharp \circ \gamma)(s)$. Near this $s$, the eigenvalue changes its sign once, and hence contributes to the spectral flow by $\pm 1$. To identify the sign of this contribution, we need to know the sign of $\det J(p)$, which determines whether $(a, b)$ is orientation preserving or not. We also need the sign of $d(a \circ f_j)/dt$. By direct inspection, the positive (resp.\ negative) intersection point contributes to the spectral flow by $+1$ (resp.\ $-1$) when the product of the sings of $\det J$ and $d(a \circ f_j)/dt$ is $+1$ (resp.\ $-1$). This product of signs agrees with $\epsilon(f_j)$, as shown in Lemma \ref{lem:function_associated_to_Fermi_arc}. To summarize, the intersection points of $f_j$ and $\gamma$ contribute to the spectral flow by $\epsilon(f_j) I(f_j, \gamma)$. Taking all the connected components of $a^{-1}((-1, 1)) \cap b^{-1}(0)$ into account, we conclude the expression of the spectral flow.
\end{proof}

\begin{dfn} \label{dfn:Fermi_arc_cycle}
Suppose $a^{-1}((-1, 1)) \cap b^{-1}(0) \neq \emptyset$. We define a singular homology $1$-cycle $f(\hat{H}^\sharp) \in Z^1(T^2, \pi(W))$ by
$$
f(\hat{H}^\sharp) = - \sum_{j = 1}^\ell \epsilon(f_j) f_j
- \sum_{j = 1}^{\ell^c} \epsilon(f^c_j) f^c_j,
$$
where $\epsilon$ is the sign given in Lemma \ref{lem:function_associated_to_Fermi_arc}. 
\end{dfn}

Note that, in the case of $\ell > 0$, Lemma \ref{lem:extension} implies $\pi(W) \neq 0$, and we use the continuous extension $f_i : [0, 1] \to T^2$ in the above definition of $f(\hat{H}^\sharp)$.

\begin{lem} \label{lem:Fermi_acr_and_spectral_flow}
If $a^{-1}((-1, 1)) \cap b^{-1}(0) \neq \emptyset$, then $f(\hat{H}^\sharp)$ represents $\mathcal{F}(\hat{H}^\sharp)$.
\end{lem}

\begin{proof}
We make the choices in Lemma \ref{lem:choice}. In this choice, we can suppose $k^0_x$ and $k^0_y$ are as in Assumption \ref{assumption} (d). Let $f$ be one of the smooth embeddings $f_1, \ldots, f_\ell, f^c_1, \ldots, f^c_{\ell^c}$. By the same argument as in the proof of Proposition \ref{prop:cohomology_basis}, we can express $\mathrm{PD}^{-1}([f]) \in H^1(T^2 \backslash \pi(W))$ as
\begin{align*}
\mathrm{PD}^{-1}([f])
&=
\bigg( \int_{\overline{\alpha}_y} \mathrm{PD}^{-1}([f]) \bigg)
\mathrm{PD}^{-1}([\overline{\alpha}_x])
- \bigg( \int_{\overline{\alpha}_x} \mathrm{PD}^{-1}([f]) \bigg)
\mathrm{PD}^{-1}([\overline{\alpha}_y]) \\
&\quad
- \sum_{i = 1}^{n-1} 
\bigg( \int_{\partial \overline{D}_i} \mathrm{PD}^{-1}([f]) \bigg)
\mathrm{PD}^{-1}([\overline{\alpha}_i]),
\end{align*}
where $D_i$, ($i = 0, 1, \ldots, n-1$) are small open disks centered at $\overline{w}_i \in \pi(W)$ such that $\overline{D}_i \cap \overline{D}_j = \emptyset$ for $i \neq j$. Taking the Poincar\'{e} duality, we get the following expression of $[f] \in H_1(T^2, \pi(W))$
\begin{align*}
[f]
&=
I(f, \overline{\alpha}_y) [\overline{\alpha}_x]
- I(f, \overline{\alpha}_x) [\overline{\alpha}_y]
- \sum_{i = 1}^{n-1} 
I(f, \partial \overline{D}_i)
[\overline{\alpha}_i],
\end{align*}
where $I(f, \gamma) = \int_\gamma \mathrm{PD}^{-1}([f])$ for any smooth map $\gamma : S^1 \to T^2 \backslash \pi(W)$. To summarize, we have
\begin{align*}
[f(\hat{H}^\sharp)]
&=
-\bigg(
\sum_{j = 1}^\ell \epsilon(f_j) I(f_j, \overline{\alpha}_y)
+ \sum_{j = 1}^{\ell^c} \epsilon(f^c_j) I(f^c_j, \overline{\alpha}_y)
\bigg) [\overline{\alpha}_x] \\
&+
\bigg(
\sum_{j = 1}^\ell \epsilon(f_j) I(f_j, \overline{\alpha}_x)
+ \sum_{j = 1}^{\ell^c} \epsilon(f^c_j) I(f^c_j, \overline{\alpha}_x)
\bigg) [\overline{\alpha}_y] \\
&+
\sum_{i = 1}^{n-1}
\bigg(
\sum_{j = 1}^\ell \epsilon(f_j) I(f_j, \partial \overline{D}_i)
+ \sum_{j = 1}^{\ell^c} \epsilon(f^c_j) I(f^c_j, \partial \overline{D}_i)
\bigg) [\overline{\alpha}_i].
\end{align*}

Because of Assumption \ref{assumption} (d), we can regard $\overline{\alpha}_y$ as a smooth embedding $\overline{\alpha}_y : S^1 \to T^2 \backslash \pi(W)$. By the transversality theorem \cite{G-P}, we can perturb $\overline{\alpha}_y$ to another smooth embedding $\gamma$ so that $\gamma$ intersects with the closed submanifold $a^{-1}((-1, 1)) \cap b^{-1}(0) \subset T^2 \backslash \pi(W)$ transversally. This means that $\gamma$ intersects with each of $f = f_1, \ldots, f_\ell, f^c_1, \ldots, f^c_{\ell^c}$ transversally. By Assumption \ref{assumption} (d) again, if there exists an intersection of $\overline{\alpha}_y$ and $a^{-1}((-1, 1)) \cap b^{-1}(0)$, then $\det J \neq 0$ at the intersection point. Thus, if the perturbation is small enough, then $\det J \neq 0$ at the intersection points of $\gamma$ and $f = f_1, \ldots, f_\ell, f^c_1, \ldots, f^c_{\ell^c}$. Noting that the spectral flow $\mathrm{sf}(\hat{H}^\sharp \circ \gamma)$ and the intersection number $I(f, \gamma)$ depend only on the homotopy class of $\gamma : S^1 \to T^2 \backslash \pi(W)$, we apply Lemma \ref{lem:spectral_flow_as_intersection_number} to have
\begin{align*}
\mathrm{sf}(\hat{H}^\sharp|_{\overline{\alpha}_y})
&= 
\mathrm{sf}(\hat{H}^\sharp|_{\gamma}) 
= \sum_{j = 1}^\ell \epsilon(f_j) I(f_j, \gamma)
+ \sum_{j = 1}^{\ell^c} \epsilon(f^c_j) I(f^c_j, \gamma) \\
&
= \sum_{j = 1}^\ell \epsilon(f_j) I(f_j, \overline{\alpha}_y)
+ \sum_{j = 1}^{\ell^c} \epsilon(f^c_j) I(f^c_j, \overline{\alpha}_y).
\end{align*}
The same argument proves
\begin{align*}
\mathrm{sf}(\hat{H}^\sharp|_{\overline{\alpha}_x})
& 
= \sum_{j = 1}^\ell \epsilon(f_j) I(f_j, \overline{\alpha}_x)
+ \sum_{j = 1}^{\ell^c} \epsilon(f^c_j) I(f^c_j, \overline{\alpha}_x).
\end{align*}
Finally, we choose and fix $i \in \{ 1, \ldots, n - 1 \}$. Let $U_{\overline{w}_i} \subset T^2$ be the open set in Assumption \ref{assumption} (b). We can choose the open disk $D_i$ so as to be $\partial \overline{D}_i \subset U_{\overline{w}_i} \backslash \{ \overline{w}_i \}$. Then, by the transversality theorem again, the smooth embedding $\partial \overline{D}_i : S^1 \to U_{\overline{w}_i} \backslash \{ \overline{w}_i \}$ can be perturbed to a smooth embedding $\gamma :  S^1 \to U_{\overline{w}_i} \backslash \{ \overline{w}_i \}$ so that $\gamma$ intersects with the closed submanifold $(U_{\overline{w}_i} \backslash \{ \overline{w}_i \}) \cap a^{-1}((-1, 1)) \cap b^{-1}(0)$ in $U_{\overline{w}_i} \backslash \{ \overline{w}_i \}$ transversally. Then $\gamma$ intersects with each of $f = f_1, \ldots, f_\ell, f^c_1, \ldots, f^c_{\ell^c}$ transversally. Since an intersection point of $\gamma$ and $a^{-1}((-1, 1)) \cap b^{-1}(0)$ is contained in $(U_{\overline{w}_i} \backslash \{ \overline{w}_i \}) \cap a^{-1}((-1, 1)) \cap b^{-1}(0)$, Assumption \ref{assumption} (b) implies $\det J \neq 0$ at the intersection point. Thus, from Lemma \ref{lem:spectral_flow_as_intersection_number} together with the homotopy invariance of the spectral flow and the intersection number, it follows that
\begin{align*}
\mathrm{sf}(\hat{H}^\sharp|_{\partial \overline{D}_i})
& 
= \sum_{j = 1}^\ell \epsilon(f_j) I(f_j, \partial \overline{D}_i)
+ \sum_{j = 1}^{\ell^c} \epsilon(f^c_j) I(f^c_j, \partial \overline{D}_i).
\end{align*}
In summary, we get
$$
[f(\hat{H}^\sharp)]
= -\mathrm{sf}(\hat{H}^\sharp|_{\overline{\alpha}_y}) [\overline{\alpha}_x]
+ \mathrm{sf}(\hat{H}^\sharp|_{\overline{\alpha}_x}) [\overline{\alpha}_y]
+ \sum_{i = 1}^{n-1} 
\mathrm{sf}(\hat{H}^\sharp|_{\partial \overline{D}_i})
[\overline{\alpha}_i].
$$
This is nothing but the expression of $\mathcal{F}(\hat{H}^\sharp)$ given in Lemma \ref{lem:spectral_flow_class_in_basis_expression}.
\end{proof}

Now, we are in the position to prove Theorem \ref{thm:main}:

\begin{proof}[Proof of Theorem \ref{thm:main}]
We let $\mathcal{F}(\hat{H}^\sharp) \in H^1(T^2, \pi(W))$ be the homology class given in Definition \ref{dfn:edge_index}. Then the first two properties (1) and (2) are satisfied. Hence we show (3) here. 

Suppose that $a^{-1}((-1, 1)) \cap b^{-1}(0)) = \emptyset$. In this case, we have $\pi(W) = F$. Thus, we assume $\pi(W) = F \neq \emptyset$. Then, $F$ consists of a finite number of points. Notice here that the continuous map $\hat{H}^\sharp : T^2 \backslash \pi(W) \to \mathrm{Fred}^1$ takes values in the subspace of invertible operators. As a result, $\mathcal{F}(\hat{H}^\sharp) = 0$. Thus, if we take constant paths with values in the finite set $\pi(W) = F$, then they represent the trivial homology class $\mathcal{F}(\hat{H}^\sharp) = 0$, and (3) is verified.

Suppose that $a^{-1}((-1, 1)) \cap b^{-1}(0)) \neq \emptyset$. In this case, we have the homology cycle $f(\hat{H}^\sharp)$ as in Definition \ref{dfn:Fermi_arc_cycle}. By construction, $F$ is the union of the images of the paths $f_1, \ldots, f_\ell, f^c_1, \ldots, f^c_{\ell^c}$ representing the homology class $[f(\hat{H}^\sharp)]$. By Lemma \ref{lem:Fermi_acr_and_spectral_flow}, we have $[f(\hat{H}^\sharp)] = \mathcal{F}(\hat{H}^\sharp)$, and (3) is verified. 
\end{proof}

\bigskip

Finally, we relate the characterization (W2) of Weyl semimetal to Assumption \ref{assumption} (b). For this aim, we understand the meaning of the approximation $\sim$ in (W2) as the agreement of the Taylor expansions up to the first order after a locally linear transformation of coordinates.

\begin{lem} \label{lem:characterization}
With the above understating, (W2) implies Assumption \ref{assumption} (b).
\end{lem}

\begin{proof}
By the conjugation of the matrix $T$ in the proof of Lemma \ref{lem:spectral_flow_of_Chern_insulator} and the exchange of coordinates $k_y \leftrightarrow k_z$, which comprise a locally linear transformation of coordinates, we have 
$$
T \hat{H}_{\mathrm{loc}}(a(k_x, k_z), b(k_x, k_z), k_y) T^*
= 
b(k_x, k_z) \sigma_x + \sin k_y \sigma_y + (a(k_x, k_z) + \cos k_z) \sigma_z.
$$
Suppose that $(k^0_x, k^0_y, k^0_z) \in W$ is a Weyl point such that $a(k^0_x, k^0_z) = -1$, $b(k^0_x, k^0_z) = 0$ and $k^0_y = 0$. At this point, if we demand from (W2) that the Taylor expansion of $T \hat{H}_{\mathrm{loc}}(a(k_x, k_z), b(k_x, k_z)) T^*$ agrees with $\sum _{i = x, y, z} (k_i - k^0_i) \sigma_i$ up to the first order, then 
\begin{align*}
\frac{\partial a}{\partial k_z}(k^0_x, k^0_z) &= 0, &
\frac{\partial a}{\partial k_z}(k^0_x, k^0_z) &= 1, \\
\frac{\partial b}{\partial k_z}(k^0_x, k^0_z) &= 1, &
\frac{\partial b}{\partial k_z}(k^0_x, k^0_z) &= 0.
\end{align*}
Thus, $\det J (k^0_x, k^0_z) \neq 0$, and hence Assumption \ref{assumption} (b) is fulfilled. At a Weyl point such that $a(k^0_x, k^0_z) = 1$, $b(k^0_x, k^0_z) = 0$ and $k^0_y = \pi$, a similar argument is valid.
\end{proof}

\subsection{Example}
\label{subsec:example}

We here give three examples of $\hat{H} : T^3 \to \mathrm{Herm}(\C^2)_0$ which satisfies Assumption \ref{assumption}.

\medskip

The first example satisfying Assumption \ref{assumption} is 
\begin{align*}
\hat{H}(k_x, k_y, k_z)
&=
\left(
\begin{array}{cc}
\sin k_y & 2 + \cos k_x + \cos k_y + e^{ik_z} \\
2 + \cos k_x + \cos k_y + e^{-ik_z} & - \sin k_y
\end{array}
\right) \\
&=
\hat{H}_{\mathrm{loc}}(a(k_x, k_y), b(k_x, k_y), k_z),
\end{align*}
where $a(k_x, k_y) = 2 + \cos k_x + \cos k_y$ and $b(k_x, k_y) = \sin k_y$. The Weyl points are
\begin{align*}
w_0^1 &= (\pi/2, \pi, 0), & w_1^1 &= (3\pi/2, \pi, \pi),
\end{align*}
and the Fermi arc consists of a closed interval
$$
F = [\pi/2, 3\pi/2] \times \{ \pi \}.
$$

Assumption \ref{assumption} (a) is clear, since $\pi(W) = a^{-1}(\{ \pm 1 \}) \cap b^{-1}(0) = \{ \overline{w}_0, \overline{w}_1 \}$, where
\begin{align*}
\overline{w}_0 &= \pi(w_0) = (\pi/2, \pi), &
\overline{w}_1 &= \pi(w_1) = (3\pi/2, \pi).
\end{align*}
The determinant of the Jacobian of $(a, b) : T^2 \to \R^2$ is computed as $\det J(k_x, k_y) = - \sin k_x \cos k_y$. Assumption \ref{assumption} (b) follows from $\det J(\overline{w}) \neq 0$ at $\overline{w} \in \pi(W)$. We here notice that $a^{-1}((-1, 1))$ is contained in $(\pi/2, 3\pi/2) \times (\pi/2, 3\pi/2)$. In this region, $b = 0$ and $\partial b/\partial k_y = -1$ on $(\pi/2, 3\pi/2) \times \{ \pi \}$, so that $0$ is a regular value of $b$ on $a^{-1}((-1, 1))$. We have $b^{-1}(0) \cap a^{-1}((-1, 1)) = (\pi/2, 3\pi/2) \times \{ \pi \}$. This is connected, and Assumption \ref{assumption} (c) is verified. Finally, Assumption \ref{assumption} (d) is satisfied by $(k^0_x, k^0_y) = (0, 0)$.

%

Now, we take $\overline{\alpha}_1 : [0, 1] \to T^2$ in Lemma \ref{lem:choice} to be $\overline{\alpha}_1(t) = ((2t + 1)\pi/2, \pi)$. Applying Lemma \ref{lem:spectral_flow_class_in_basis_expression}, we have
\begin{align*}
\mathcal{F}(\hat{H}^\sharp)
&=
-
\overbrace{\mathrm{sf}(\hat{H}^\sharp|_{\{ k^0_x\} \times S^1})}^{0} 
[\overline{\alpha}_x]
+ 
\overbrace{\mathrm{sf}(\hat{H}^\sharp|_{S^1 \times \{ k^0_y\}})}^{0} 
[\overline{\alpha}_y]
+ 
\overbrace{\mathrm{sf}(\hat{H}^\sharp|_{\partial \overline{D}_1})}^{1} 
[\overline{\alpha}_1]
=
[\overline{\alpha}_1].
\end{align*}
We can take an embedding $f_1 : (0, 1) \to T^2$ corresponding to the submanifold $a^{-1}((-1, 1)) \cap b^{-1}(0) = F \backslash \pi(W)$ as $f_1(t) = ((2t + 1)\pi/2, \pi)$, so that its extension $f_1 : [0, 1] \to T^2$ agrees with $\overline{\alpha}_1$. We readily see $\epsilon(f_1) = -1$, and hence $f(\hat{H}^\sharp) = -\epsilon(f_1) f_1 = \overline{\alpha}_1$, as anticipated. By construction, $\mathrm{Im}f_1 = F$.

Note that we can take another choice of $(k^0_x, k^0_y)$. For example, we can choose ${}'k^0_x \in (\pi/2, 3\pi/2) \backslash \{ \pi \}$ instead of $k^0_x = 0$. Then we have ${}'\overline{\alpha}_y$ correspondingly. The same choice $k^0_y = 0$ as before leads to the same corresponding cycle $\overline{\alpha}_x$. Then, to meet the conditions in Lemma \ref{lem:choice}, we take ${}'\overline{\alpha}_1 : [0, 1] \to T^2$ to be ${}'\overline{\alpha}_1(t) = ((-2t + 1)\pi/2, \pi)$. Under these choices, Lemma \ref{lem:spectral_flow_class_in_basis_expression} leads to 
\begin{align*}
\mathcal{F}(\hat{H}^\sharp)
&=
-
\overbrace{\mathrm{sf}(\hat{H}^\sharp|_{\{ {}'k^0_x\} \times S^1})}^{-1} 
[\overline{\alpha}_x]
+ 
\overbrace{\mathrm{sf}(\hat{H}^\sharp|_{S^1 \times \{ k^0_y\}})}^{0} 
[{}'\overline{\alpha}_y]
+ 
\overbrace{\mathrm{sf}(\hat{H}^\sharp|_{\partial \overline{D}_1})}^{1} 
[{}'\overline{\alpha}_1] \\
&= [\overline{\alpha}_x] + [{}'\overline{\alpha}_1].
\end{align*}
This is consistent with the former calculation, since $[\overline{\alpha}_1] = [\overline{\alpha}_x] + [{}'\overline{\alpha}_1]$ in the homology group $H_1(T^2, \pi(W)) \cong \Z^3$.

\bigskip

The second example is
\begin{align*}
\hat{H}(k_x, k_y, k_z)
&=
\left(
\begin{array}{cc}
\sin k_y & \cos k_x + \cos k_y + e^{ik_z} \\
\cos k_x + \cos k_y + e^{-ik_z} & - \sin k_y
\end{array}
\right) \\
&=
\hat{H}_{\mathrm{loc}}(a(k_x, k_y), b(k_x, k_y), k_z),
\end{align*}
where $a(k_x, k_y) = \cos k_x + \cos k_y$ and $b(k_x, k_y) = \sin k_y$. The Weyl points are
\begin{align*}
w_0^1 &= (\pi/2, 0, \pi), & w_1^1 &= (3\pi/2, 0, \pi), &
w_2^1 &= (\pi/2, \pi, 0), & w_3^1 &= (3\pi/2, \pi, 0).
\end{align*}
Their projections onto $T^2$ are
\begin{align*}
\overline{w}_0 &= (\pi/2, 0), & \overline{w}_1 &= (3\pi/2, 0), &
\overline{w}_2 &= (\pi/2, \pi), & \overline{w}_3 &= (3\pi/2, \pi).
\end{align*}
The Fermi arc consists of disjoint two closed intervals
$$
F = [\pi/2, 3\pi/2] \times \{ 0 \} \cup [-\pi/2, \pi/2] \times \{ \pi \}.
$$
We choose $(k^0_x, k^0_y) = (\pi/4, 7\pi/4)$, and
\begin{align*}
\overline{\alpha}_1(t) &= ((2t + 1)\pi/2, 0), &
\overline{\alpha}_2(t) &= (\pi/2, \pi t), &
\overline{\alpha}_3(t) &= ((2t + 1)\pi/2, \pi t).
\end{align*}
Then we compute the spectral flows to have
\begin{align*}
\mathcal{F}(\hat{H}^\sharp)
&= 
-
\overbrace{\mathrm{sf}(\hat{H}^\sharp|_{\{ k^0_x\} \times S^1})}^{-1} 
[\overline{\alpha}_x]
+ 
\overbrace{\mathrm{sf}(\hat{H}^\sharp|_{S^1 \times \{ k^0_y\}})}^{0} 
[\overline{\alpha}_y] \\
&\quad
+ \overbrace{\mathrm{sf}(\hat{H}^\sharp|_{\partial \overline{D}_1})}^{-1} 
[\overline{\alpha}_1]
+ \overbrace{\mathrm{sf}(\hat{H}^\sharp|_{\partial \overline{D}_2})}^{1}
[\overline{\alpha}_2]
+ \overbrace{\mathrm{sf}(\hat{H}^\sharp|_{\partial \overline{D}_3})}^{-1}
[\overline{\alpha}_3] \\
&= [\overline{\alpha}_x] 
- [\overline{\alpha}_1] + [\overline{\alpha}_2] - [\overline{\alpha}_3].
\end{align*}
We can take the embeddings $f_1, f_2 : (0, 1) \to T^2$ corresponding to the connected components of the submanifold $a^{-1}((-1, 1)) \cap b^{-1}(0) = F \backslash \pi(W)$ as follows
\begin{align*}
f_1 &= ((2t + 1)\pi/2, 0), &
f_2 &= ((-2t + 1)\pi /2, \pi).
\end{align*}
Using the continuous extensions $f_1, f_2 : [0, 1] \to T^2$, we have 
$$
f(\hat{H}^\sharp) = - \epsilon(f_1) f_1 - \epsilon(f_2) f_2 
= - f_1 - f_2.
$$
This is consistent with expression of $\mathcal{F}(\hat{H}^\sharp)$, since $f_1 = \alpha_1$ and $- f_2 = \overline{\alpha}_x + \overline{\alpha}_2 - \overline{\alpha}_3$ in the relative homology group $H_1(T^2, \pi(W)) \cong \Z^5$.

%
%
%
%
%
%

\bigskip

The third example is 
\begin{align*}
\hat{H}(k_x, k_y, k_z)
&=
\left(
\begin{array}{cc}
\sin k_y - \cos k_x & 2 + \cos k_x + \cos k_y + e^{ik_z} \\
2 + \cos k_x + \cos k_y + e^{-ik_z} & - \sin k_y + \cos k_x
\end{array}
\right) \\
&=
\hat{H}_{\mathrm{loc}}(a(k_x, k_y), b(k_x, k_y), k_z),
\end{align*}
where $a(k_x, k_y) = 2 + \cos k_x + \cos k_y$ and $b(k_x, k_y) = \sin k_y - \cos k_x$. Note that the determinant of the Jacobian of $(a, b) : T^2 \to \R^2$ is
$$
\det J(k_x, k_y) =
- \sin k_x (\cos k_y - \sin k_y)
= - \sqrt{2} \sin k_x \cos (k_y + \pi/4).
$$
The Weyl points are
\begin{align*}
w_0^1 &= (\pi/2, \pi, \pi), &
w_1^1 &= (\pi, 2\pi/2, \pi), &
w_2^1 &= (3\pi/2, \pi, \pi).
\end{align*}
Their projections onto $T^2$ are
\begin{align*}
\overline{w}_0 &= (\pi/2, \pi), &
\overline{w}_1 &= (\pi, 3\pi/2), &
\overline{w}_2 &= (3\pi/2, \pi).
\end{align*}
The Fermi arc is the union of two closed intervals, and is described as
$$
F = \mathrm{Im}(f_1) \cup \mathrm{Im}(f_2),
$$
where $f_1, f_2 : [0, 1] \to T^2$ are given by
\begin{align*}
f_1(t) &= ((t + 1)\pi/2, (t + 2)\pi/2), &
f_2(t) &= ((t + 2)\pi/2, (-t + 3)\pi/2).
\end{align*}
We choose $(k^0_x, k^0_y) = (0, 0)$, and
\begin{align*}
\overline{\alpha}_1(t) &= ((t + 1)\pi/2, (t+2)\pi/2), &
\overline{\alpha}_2(t) &= ((2t + 1)\pi/2, \pi).
\end{align*}
We compute the spectral flows to get
\begin{align*}
\mathcal{F}(\hat{H}^\sharp)
&= 
-
\overbrace{\mathrm{sf}(\hat{H}^\sharp|_{\{ k^0_x\} \times S^1})}^{0} 
[\overline{\alpha}_x]
+ 
\overbrace{\mathrm{sf}(\hat{H}^\sharp|_{S^1 \times \{ k^0_y\}})}^{0} 
[\overline{\alpha}_y] \\
&\quad
+ \overbrace{\mathrm{sf}(\hat{H}^\sharp|_{\partial \overline{D}_1})}^{0} 
[\overline{\alpha}_1]
+ \overbrace{\mathrm{sf}(\hat{H}^\sharp|_{\partial \overline{D}_2})}^{1}
[\overline{\alpha}_2] \\
&= 
[\overline{\alpha}_2].
\end{align*}
The maps $f_1, f_2$ are the continuous extensions of the embeddings corresponding to the connected components of the submanifold $a^{-1}((-1, 1)) \cap b^{-1}(0)$. We have $\epsilon(f_1) = \epsilon(f_2) = -1$ and
$$
f(\hat{H}^\sharp) = -\epsilon(f_1) f_1 - \epsilon(f_2) f_2 = f_1 + f_2.
$$
This represents $\mathcal{F}(\hat{H}^\sharp)$, since $\overline{\alpha}_2 = f_1 + f_2$ in $H_1(T^2, \pi(W)) \cong \Z^4$.


\section{Spectral analysis of the local model}
\label{sec:spectral_analysis}

In this section, we prove Proposition \ref{prop:spectral_data}. Let $\hat{H} : S^1 \to \mathrm{Herm}(\C^2)_0$ be the continuous map given by
$$
\hat{H}(\theta) = \hat{H}_{\mathrm{loc}}(a, b, \theta)=
\left(
\begin{array}{cc}
b & a + e^{i\theta} \\
a + e^{-i\theta} & -b
\end{array}
\right),
$$
where $a, b \in \R$. Let $H : L^2(\Z, \C^2) \to L^2(\Z, \C^2)$ be the Fourier transformation of the multiplication operator with $\hat{H}$. Concretely, $H$ acts on $\psi = (\psi(n))_{n \in \Z}$ by
$$
(H\psi)(n) = A \psi(n-1) + V \psi(n) + A^*\psi(n+1),
$$
in which the square matrices $A$ and $V$ are
\begin{align*}
A &= 
\left(
\begin{array}{rr}
0 & 1 \\
0 & 0 
\end{array}
\right),
&
V &=
\left(
\begin{array}{rr}
b & a \\
a & -b 
\end{array}
\right).
\end{align*}
Define $H^\sharp : L^2(\N, \C^2) \to L^2(\N, \C^2)$ by $H^\sharp = P H P^*$, where $P : L^2(\Z, \C^2) \to L^2(\N, \C^2)$ is the orthogonal projection.

\medskip

Given $E \in \R$, we are to consider the solutions to the equation $(H^\sharp - E)\psi = 0$, which is equivalent to the equations for $\psi(1), \psi(2), \psi(3), \ldots \in \C^2$
\begin{align*}
0 &= (V - E) \psi(1) + A^*\psi(2), \\
0 &= A \psi(1) + (V - E) \psi(2) + A^*\psi(3), \\
0 &= A \psi(2) + (V - E) \psi(3) + A^*\psi(4), \\
& \vdots \\
0 &= A \psi(n-1) + (V(k_x) - E) \psi(n) + A^*\psi(n+1).
\end{align*}

\begin{lem} 
Suppose that $a = 0$. 
\begin{itemize}
\item
In the case that $E = \pm \sqrt{b^2 + 1}$, the equation $(H^\sharp - E)\psi = 0$ admits $L^2$-solutions. The solutions form a vector space of infinite dimension.

\item
In the case that $E = b$, the equation $(H^\sharp - E)\psi = 0$ admits $L^2$-solutions. The solutions form a vector space of $1$-dimension.

\item
In the other cases, the equation $(H^\sharp - E)\psi = 0$ admits no non-trivial $L^2$-solutions.

\end{itemize}
\end{lem}

\begin{proof}
We express $\psi(n) = {}^t(\alpha(n) \ \beta(n))$. The equation $0 = (V - E) \psi(1) + A^*\psi(2)$ is equivalent to
\begin{align*}
(b - E) \alpha(1) &= 0, &
\alpha(2) &= (b + E)\beta(1).
\end{align*}
For $n \ge 1$, the equation $0 = A \psi(n) + (V - E) \psi(n+1) + A^*\psi(n+2)$ is equivalent to
\begin{align*}
0 &= \beta(n) + (b - E) \alpha(n+1), &
0 &= - (b + E) \beta(n+1) + \alpha(n+2).
\end{align*}
In summary, we have $(b - E) \alpha(1) = 0$ and
\begin{align*}
\alpha(n + 1) &= (b + E) \beta(n), &
(b - E) \alpha(n +1) &= - \beta(n).
\end{align*}
for $n \ge 1$. Suppose here that $b - E \neq 0$. Then the first equation gives $\alpha(1) = 0$, and $\beta(1) \in \C$ can be any complex number. For $n \ge 1$, we have
$$
\alpha(n+1) 
=  (b + E) \beta(n)
= \frac{-1}{b - E} \beta(n).
$$
This implies that $E^2 = b^2 + 1$, or equivalently $E = \pm \sqrt{b^2 + 1}$, which is consistent with $b - E \neq 0$. Under $E^2 = b^2 + 1$, we have
$$
\alpha(n+1) = (b + E) \beta(n)
$$
for $n \ge 1$. Thus, a solution is given by $\alpha(1) = 0$ and $\alpha(n + 1) = (b + E) \beta(n)$ for $n \ge 1$, where $\beta(1), \beta(2), \ldots \in \C$ are any complex numbers. The square of the $L^2$-norm of this solution is
$$
\lVert \psi \rVert^2
= (1 + \lvert E + b \rvert^2) \sum_{n \ge 1} \vert \beta(n) \rvert^2.
$$
Thus, for each $\beta = (\beta(n))_{n \in \N}$ such that $\sum_{n \ge 1} \lvert \beta(n) \vert^2 < \infty$, we get an $L^2$-solution to $(H^\sharp - E)\psi = 0$. The space of solutions is apparently infinite dimensional. To complete the proof, we suppose that $b - E = 0$. In this case, we can readily solve the equations. A solution is specified by a complex number $\alpha(1) \in \C$ and $\alpha(n+1) = \beta(n) = 0$ for $n \ge 1$. It is clear that this solution gives rise to an $L^2$-solution to $(H^\sharp - E)\psi = 0$, and the space of solutions is $1$-dimensional. 
\end{proof}

\begin{lem}
Suppose that $a \neq 0$. If we ignore the $L^2$-condition, then any solution $\psi = (\psi(n))_{n \in \N}$ to $(H^\sharp - E)\psi = 0$ can be described as follows
\begin{align*}
\psi(1)
&=
z
\left(
\begin{array}{c}
a \\ E- b
\end{array}
\right),
&
\psi(n+1)
&=
R \psi(n), \quad (n \ge 1)
\end{align*}
where $z \in \C$ is a complex number, and 
$$
R =
\left(
\begin{array}{cc}
-a & b + E \\
b - E & \frac{E^2 - b^2 - 1}{a}
\end{array}
\right).
$$
\end{lem}

\begin{proof}
We express $\psi(n) = {}^t(\alpha(n) \ \beta(n))$. The equation $0 = (V - E) \psi(1) + A^*\psi(2)$ is equivalent to
\begin{align*}
(b - E) \alpha(1) + a \beta(1) & = 0, &
a \alpha(1) - (b + E)\beta(1) + \alpha(2) &= 0.
\end{align*}
Under the assumption $a \neq 0$, these equations lead to
\begin{align*}
\alpha(2) &= - a \alpha(1) + (b + E) \beta(1), &
\beta(1) &= \frac{E - b}{a} \alpha(1), &
\end{align*}
For $n \ge 1$, the equation $0 = A \psi(n) + (V - E) \psi(n+1) + A^*\psi(n+2)$ leads to
\begin{align*}
0 &= \beta(n) + (b - E) \alpha(n+1) + a \beta(n+1), \\
0 &= a \alpha(n+1) - (b + E) \beta(n+1) + \alpha(n+2).
\end{align*}
To summarize, $\alpha(1)$ determines $\beta(1)$ as $\beta(1) = \frac{E - b}{a} \alpha(1)$, and we have
\begin{align*}
\alpha(n + 1) &= -a \alpha(n) + (b + E)\alpha(n), &
\beta(n + 1) &= \frac{E - b}{a} \alpha(n+1) - \frac{1}{a} \beta(n),
\end{align*}
for $n \ge 1$. This provides us
\begin{align*}
\alpha(n + 1) &= - a \alpha(n) + (b + E)\beta(n), \\
\beta(n+1) &= (b - E) \alpha(n) + \frac{E^2 - b^2 - 1}{a} \beta(n),
\end{align*}
for $n \ge 1$. Hence the proof is completed. 
\end{proof}

The trace and the determinant of $R$ are
\begin{align*}
\mathrm{tr}R &= \frac{E^2 - a^2 - b^2 - 1}{a}, &
\det R &= 1.
\end{align*}
The discriminant of the quadratic equation $\lambda^2 - \lambda \mathrm{tr}R + \det R = 0$ is
$$
\Delta = (\mathrm{tr}R)^2 - 4 \det R 
= \frac{((a + 1)^2 - (E^2 - b^2))((a - 1)^2 - (E^2 - b^2))}{a^2}.
$$

\begin{lem} \label{lem:eigenvalue_non_positive_discriminant}
Suppose that $a \neq 0$. In this case, if $\Delta \le 0$, then there is no $L^2$-solutions to $(H^\sharp - E) \psi = 0$.
\end{lem}

\begin{proof}
In the case that $\Delta < 0$, the matrix $R$ has distinct eigenvalues $\lambda$ and $\bar{\lambda}$ such that $\lambda \overline{\lambda} = 1$. We diagonalize $R$ as $T^{-1} R T = \mathrm{diag}(\lambda, \lambda')$, with $\lvert \lambda \rvert = \lvert \lambda' \vert = 1$. Then the $L^2$-norm of a solution $\psi$ to $(H^\sharp - E) \psi = 0$ is convergent if and only if that of $T^{-1}\psi = (T^{-1}\psi(n))$ is convergent. If we express $T^{-1}\psi(1) = {}^t(\alpha'(1) \ \beta'(1))$, then 
$$
\lVert T^{-1} \psi \rVert^2
= \sum_{n \ge 1} (
\lvert \lambda \rvert^{2n} \lvert \alpha'(1) \rvert^2 +
\lvert \lambda' \rvert^{2n} \lvert \beta'(1) \rvert^2 
)
=
\sum_{n \ge 1} (
\lvert \alpha'(1) \rvert^2 +
\lvert \beta'(1) \rvert^2 
),
$$
which is divergent unless $\psi = 0$. In the case that $\Delta = 0$, the matrix $R$ has an eigenvalue $\lambda = 1$ with multiplicity $2$, or $\lambda = -1$ with multiplicity $2$. When $R$ is diagonalizable, we have $R = \pm 1_{\C^2}$, so that the $L^2$-norm of $\psi$ is divergent unless $\psi = 0$. When $R$ is not diagonalizable, we consider the Jordan normal form of $R$: In the case that $\lambda = 1$, there is an invertible matrix $T$ such that
$$
T^{-1}RT =
\left(
\begin{array}{cc}
1 & 1 \\
0 & 1
\end{array}
\right).
$$
As before, unless $\psi = 0$, the $L^2$-norm of $T^{-1}\psi$ is divergent, and so is $\lVert \psi \rVert$. The same holds true in the case that $\lambda = -1$.
\end{proof}

\begin{lem}
Suppose that $a \neq 0$, $\Delta > 0$ and $E = -b$. In this case, $\Delta > 0$ is equivalent to $a \neq \pm 1$, and the equation $(H^\sharp - E)\psi = 0$ admits $L^2$-solutions if and only if $E = -b = 0$ and $\lvert a \rvert < 1$. The solutions form a vector space of dimension $1$. 
\end{lem}

\begin{proof}
Under the assumption, the eigenvalues of $R$ are solutions to
$$
\lambda^2 - 
\left( -a - \frac{1}{a} \right) \lambda + 1 = 0,
$$
so that $\lambda = -a, -1/a$, where $a \neq \pm 1, 0$. Let $T$ be the matrix
$$
T =
\left(
\begin{array}{cc}
\frac{1}{a} - a & 0 \\
2b & 1
\end{array}
\right).
$$
This is invertible, and diagonalizes $R$ as $T^{-1} R T = \mathrm{diag}(-a, -1/a)$. We have
$$
\psi(n) = R^n \psi(1)
= \frac{z}{\det T}
T
\left(
\begin{array}{c}
-(-a)^{n+1} \\
2b\big( - \frac{1}{a} \big)^{n+1}
\end{array}
\right).
$$
Therefore $\lVert \psi \rVert < +\infty$ if and only if $\lvert a \rvert < 1$ and $b = 0$. 
\end{proof}

\begin{lem}
Suppose that $a \neq 0$, $\Delta > 0$ and $E \neq -b$. In this case, the equation $(H^\sharp - E)\psi = 0$ admits $L^2$-solutions if and only if $E = b$ and $\lvert a \rvert < 1$. The solutions form a vector space of dimension $1$. 
\end{lem}

\begin{proof}
Let $\lambda_\pm$ be the eigenvalues of $R$. Let $T$ be the matrix
$$
T 
=
\left(
\begin{array}{cc}
b + E & b + E \\
a + \lambda_+ & a + \lambda_- \\
\end{array}
\right).
$$
This is invertible, since $\det T = (b + E)(\lambda_- - \lambda_+)$ does not vanish under the present assumptions $E \neq -b$ and $\Delta > 0$. The matrix $T$ diagonalizes $R$ as $T^{-1}RT = \mathrm{diag}(\lambda_+, \lambda_-)$. The $L^2$-norm of $\psi$ is convergent if and only if that of $T^{-1}\psi$ is convergent. We have
$$
\lVert T^{-1}\psi \rVert^2
= \sum_{n \ge 1}
(
\lvert \lambda_+ \rvert^{2n} \lvert \alpha'(1) \rvert^2 +
\lvert \lambda_- \rvert^{2n} \lvert \beta'(1) \rvert^2
),
$$
where $\alpha'(1)$ and $\beta'(1)$ are components of $T^{-1}\psi(1)$
$$
T^{-1}\psi(1)
=
\left(
\begin{array}{c}
\alpha'(1) \\ \beta'(1)
\end{array}
\right)
=
\frac{z}{\det T}
\left(
\begin{array}{c}
-E^2 + a^2 + b^2 + a \lambda_- \\
E^2 - a^2 - b^2 - a \lambda_+
\end{array}
\right).
$$
It follows that the equation $(H^\sharp - E)\psi = 0$ admits non-trivial $L^2$-solutions only in the following two cases:
\begin{itemize}
\item[(i)]
$\lvert \lambda_+ \rvert > 1$, $\lvert \lambda_- \rvert < 1$ and $E^2 - a^2 - b^2 - a \lambda_- = 0$; or

\item[(ii)]
$\lvert \lambda_+ \rvert < 1$, $\lvert \lambda_- \rvert > 1$ and $E^2 - a^2 - b^2 - a \lambda_+ = 0$.

\end{itemize}
In the first case, we have $\lambda_- = (E^2 - a^2 - b^2)/a$ and $\lambda_+ = 1/\lambda_-$. From the consistency $\lambda_+ + \lambda_- = \mathrm{Tr} R$, it follows that $E^2 = b^2$. Because of the assumption $E \neq -b$, we get $E = b$. In the second case, the same argument leads to $E = b$. Thus, in both cases, we have $E = b$. Now, we verify again the $L^2$-condition under the assumption $E = b$. Under this assumption, we can take $\lambda_+ = -a$ and $\lambda_- = -1/a$. We have 
$$
\left(
\begin{array}{c}
\alpha'(1) \\ \beta'(1)
\end{array}
\right)
=
\frac{z}{\det T}
\left(
\begin{array}{c}
a^2 + a \lambda_- \\
-a^2 - a \lambda_+
\end{array}
\right)
=
\frac{z}{\det T}
\left(
\begin{array}{c}
a^2 - 1 \\
0
\end{array}
\right).
$$
Note that $a^2 - 1 \neq 0$ under the present assumption $\Delta > 0$, which is equivalent to $a^2 \neq 1$. Hence the case (i) cannot occur: if $\lvert \lambda_+ \rvert = \lvert a \rvert > 1$, then there is only the trivial $L^2$-solution ($z = 0$). If $\lvert \lambda_+ \rvert = \lvert a \rvert < 1$ (the case (ii)), then we have non-trivial $L^2$-solutions which form a $1$-dimensional space. 
\end{proof}

Now, we are in the position to complete the proof of Proposition \ref{prop:spectral_data}:

\begin{prop} 
The equation $(H^\sharp - E)\psi = 0$ for $\psi = (\psi(n))_{n \in \N}$ admits a non-trivial $L^2$-solution in the following cases.
\begin{itemize}
\item
The case that $a = 0$ and $E = \pm \sqrt{b^2 + 1}$. In this case, the space of $L^2$-solutions is infinite dimensional.

\item
The case that $\lvert a \rvert < 1$ and $E = b$. In this case, the space of $L^2$-solutions is $1$-dimensional.

\end{itemize}
\end{prop}

\begin{proof}
By the lemmas so far, we have non-trivial $L^2$-solutions in the following cases.
\begin{enumerate}
\item
$a = 0$ and $E = \pm \sqrt{b^2 + 1}$. 

\item
$a = 0$ and $E = b$. 

\item
$a \neq 0$, $\lvert a \rvert < 1$ and $E = b = 0$.

\item
$a \neq 0$, $\lvert a \rvert < 1$ and $E = b \neq 0$.

\end{enumerate}
In the first case, the space of $L^2$-solutions is infinite dimensional, whereas the space is $1$-dimensional in each of the remaining cases. These remaining cases can be summarized as a single case where $\lvert a \rvert < 1$ and $E = b$.
\end{proof}


\end{document}